\documentclass[graybox]{svmult}

\usepackage{type1cm}        
\usepackage{makeidx}         
\usepackage{graphicx}        
\usepackage{multicol}        
\usepackage[bottom]{footmisc}

\usepackage[colorinlistoftodos]{todonotes}

\usepackage{hyperref}

\usepackage{newtxtext}       %
\usepackage[varvw]{newtxmath}       

\makeindex             

\newcommand{\bZ}{{\mathbb{Z}}}
\newcommand{\bQ}{{\mathbb{Q}}}

\newcommand{\bN}{{\mathbb{N}}}

\newcommand{\vf}{{\bf f}}

\newcommand{\lc}{\operatorname{lc}}

\newcommand{\im}{\operatorname{im}}

\newcommand{\id}{\operatorname{id}}
\newcommand{\ind}{\operatorname{ind}}

\newcommand{\den}{\operatorname{den}}
\newcommand{\num}{\operatorname{num}}

\newcommand{\polypart}{\operatorname{poly}}
\newcommand{\proppart}{\operatorname{proper}}

\newtheorem{thm}{Theorem}

\newtheorem{convention}[thm]{Convention}

\begin{document}

\title*{Telescoping Algorithms for $\Sigma^*$-Extensions via Complete Reductions}
\author{Shaoshi Chen\orcidID{0000-0001-8756-3006
	}, \\ Yiman Gao \thanks{Corresponding author}\orcidID{0009-0005-2070-6232}, \\ Hui Huang\orcidID{0000-0003-1067-1564} and\\ Carsten Schneider\orcidID{0000-0002-5703-4530}}
\institute{Shaoshi Chen \at KLMM, AMSS, Chinese Academy of Sciences, Beijing, China\\  \email{schen@amss.ac.cn}
\and Yiman Gao \at Johannes Kepler University Linz, Research Institute for Symbolic Computation (RISC), Linz, Austria \\ \email{ymgao@risc.jku.at} 
\and Hui Huang \at School of Mathematics and Statistics, Fuzhou University, Fuzhou, China \\ \email{huanghui@fzu.edu.cn}
\and Carsten Schneider \at Johannes Kepler University Linz, Research Institute for Symbolic Computation (RISC), Linz, Austria \\ \email{Carsten.Schneider@risc.jku.at}
}

%
%
\maketitle

\vspace*{-8cm}\begin{flushleft}
	RISC Report Series  25-05 
\end{flushleft}

\vspace*{8cm}

\abstract{A complete reduction on a difference field is a linear operator that enables one to decompose an element of the field as the sum of a summable part and a remainder such that
the given element is summable if and only if the remainder is equal to zero.
In this paper, we present a complete reduction  in a tower of $\Sigma^*$-extensions that turns to a new efficient framework for the parameterized telescoping problem. Special instances of such $\Sigma^*$-extensions cover iterative sums such as the harmonic numbers and generalized versions that arise, e.g., in combinatorics, computer science or particle physics. Moreover, we illustrate how these new ideas can be used to reduce the depth of the given sum and provide structural theorems that connect complete reductions to Karr's Fundamental Theorem of symbolic summation.
}

\section{Introduction}

The telescoping problem is a fundamental paradigm in symbolic summation. 
Given a sequence $f(k)$ that belongs to some domain of sequences, decide constructively whether there is a $g(k)$ in the same domain or a suitable extension such that
\begin{equation}\label{EQ:telescoping} 
f(k)=g(k+1)-g(k).
\end{equation}
If such a solution can be derived, one obtains the identity
$$\sum_{k=a}^{b}f(k)=g(b+1)-g(a)$$
by choosing an appropriate non-negative integer $a$. The telescoping problem has been solved for various classes of functions including rational functions \cite{Abra1971}, hypergeometric terms\cite{Gosp1978,PS1995}, and $q$-hypergeometric terms \cite{Koor1993, PR1997}. Karr \cite{Karr1981,Karr1985} introduced  so-called $\Pi\Sigma^*$-fields covering big classes of indefinite nested sums and products, and provided an algorithm to solve the telescoping problem for a given $f(k)$ that can be defined in a $\Pi\Sigma^*$-field. 
Later Karr's algorithm was improved to refined (parameterized) telescoping algorithms
that can find sum representations with optimal nesting depth~\cite{Schn2008,Schn2007,Schneider:2015,Schneider:2023} in $\Pi\Sigma^*$-fields. Furthermore, the difference field setting has been generalized to ring versions~\cite{Schn2016a,Schn2017} involving algebraic objects like $(-1)^k$ in which one can represent any indefinite nested sums defined over \hbox{($q$--)}hypergeometric products fully algorithmically. Since also the parameterized telescoping problem, a special case of Zeilberger's creative telescoping paradigm~\cite{Zeil1990b,Schn2001,PS2003} and recurrence solving~\cite{ABPS2021} can be solved in such difference fields and rings, a flexible summation machinery implemented in the Mathematica package {\tt Sigma} is available for practical problem solving. For details on all these refinements and improvements implemented in the Mathematica package {\tt Sigma} we refer to~\cite{Schn2006,Schn2021}.

An alternative method to handle the telescoping problem is to split the given summand into a summable part and a non-summable part which is minimal in some sense. The given summand is summable if and only if the latter is zero. Such a decomposition is called an additive decomposition. Additive decompositions are given for rational functions \cite{Abra1975,Paul1995}, hypergeometric terms and their $q$-analogues \cite{CHKL2015,DHL2018} without solving any difference equations explicitly. They act as preprocessors in reduction-based creative telescoping. Compared with Zeilberger's algorithm \cite{Zeil1990b,Zeil1991,PWZ1996} for computing telescopers, reduction-based method can separate the computation of telescopers and certificates. Since the size of certificates are much bigger than that of telescopers, this may lead to significant speed-ups.

A refined construction of an additive decomposition can be given in terms of linear algebra as follows~\cite[Definition 5.67]{Kaue2023}.

\begin{definition}\label{Def:CompleteReductionVS}
Let $C$ be a field, $V$ a $C$-linear space and $U$ a $C$-subspace of $V$. A $C$-linear operator $\phi:V\to V$ is called a {\em complete reduction} for $U$ if $f-\phi(f)\in U$ for all $f\in V$ and $\ker(\phi)=U$.
\end{definition}

\noindent Note that $\phi$ induces the direct sum $V=U \oplus\im(\phi)$. 
In this way, a complete reduction can be used to establish an additive decomposition ~\cite[\S 1.2]{GaoPhD2024} by choosing appropriate $C$-bases of $U$ and $\im(\phi)$ where $\im(\phi)$ contains the non-integrable/non-summable contributions.
 However, additive decompositions (as mentioned above) are usually not closed under addition, i.e., the set of minimal non-integrable/non-summable elements do not form a subspace of $V$. This refined construction appeared recently in symbolic integration.
More precisely, complete reductions have been developed for hyperexponential functions \cite{BCCLX2013}, algebraic functions \cite{CKK2016,CDK2021}, fractions of differential polynomials \cite{BLLRR2016},
Fuchsian D-finite functions \cite{CHKK2018}
and D-finite functions \cite{BCPS2018,vdHJ2021,CDK2023}. Recently, in order to compute elementary integrals effectively and discuss how to construct telescopers for non-D-finite functions, complete reductions are constructed for derivatives in exponential and primitive towers \cite{GaoPhD2024,DGLL2025}.

A natural question is how these complete reductions can be carried over from the differential case to the difference case, i.e., to a difference field $(F, \sigma)$ where $F$ is a field,  $\sigma:F\to F$ is a field automorphism and $C=\{c\in F\mid \sigma(c)=c\}$ is the field of constants. In the following we denote $\{\sigma(f)-f\mid f\in F\}$ by $\Delta(F)$ which forms a $C$-linear subspace of $F$. So there is a complementary space $R$ such that $F=\Delta(F) \oplus R$ where $R$ contains the non-summable elements in $F$. Usually, $R$ is infinite-dimensional. Hence, to obtain algorithmically a complete reduction $\phi:F\to F$ for $\Delta(F)$, we are faced with two tasks:
\begin{itemize}
\item [(1)] ~fix a complementary space $R$;
\item [(2)] ~develop an algorithm that solves the following problem: given $f\in F$, find $g\in F$ and $r\in R$ such that $f=\Delta(g)+r$.
\end{itemize}

Suppose that we can construct such a complete reduction with $(g,r)\in F^2$ and can interpret $f,r,g$ as sequences $f(k)$, $r(k)$ and $g(k)$ in terms of summation objects. Then this yields the refined telescoping equation
\begin{equation}\label{Equ:RefinedTele}
f(k)=g(k+1)-g(k)+r(k).
\end{equation}
Hence summing this equation over $k$ from $a$ to $b$ yields
\begin{equation}\label{Equ:RefinedTeleSummed}
	\sum_{k=a}^bf(k)=g(b+1)-g(a)+\sum_{k=a}^b r(k).
\end{equation}
We note further that $r\notin\Delta(F)$, i.e., the sum on the right-hand side of~\eqref{Equ:RefinedTele} is not summable in $F$. 

Some first steps related to this problem have been elaborated for $\Pi\Sigma^*$-extensions in~\cite{Schn2007,Schneider:2023} where an additive decomposition (but not a complete reduction) has been elaborated.  There we could show that $r$ is an element in the smallest possible subfield (w.r.t.\  the order of the generators producing the field $F$). Furthermore the degree of the top-most generator on which $r$ depends is optimized.
In this article we will substantially improve and refine these constructions for difference fields $(F,\sigma)$ that are built by a tower of $\Sigma^*$-extensions. More precisely, we will provide a complete reduction method that solves the telescoping problem and more generally the parameterized telescoping algorithm in $\Sigma^*$-extensions, covering creative telescoping as a special case.
Typical examples of such $\Sigma^*$-towers are difference fields generated by harmonic numbers and their generalized harmonic sums~\cite{Bluemlein:99,ABS:2011}. They play a key role in combinatorics\cite{Knut1997,Schneider:2021}, number theory \cite{boyd1994,PS2003} and particle physics \cite{BMSS:22b}. A key feature of our new reduction based method is that one does not have to solve any linear difference equation by denominator/degree boundings and rather involved linear system solving~\cite{Karr1981,Bron2000,Schn2001}. In particular, one finds not only optimal representations for $r$ in~\eqref{Equ:RefinedTeleSummed} w.r.t.\ degrees in the top variable (as derived in~\cite{Schn2007}) but also optimal representations of the coefficients of the polynomials given in the subfields.
As a consequence, our advanced algorithm can outperform the existing implementation of the summation package \texttt{Sigma} for various classes of inputs concerning simplification and time complexity aspects.

The article is organized as follows. In Section \ref{SECT:pre}, we present some basic notions related to symbolic summation and review some dual techniques that are only relevant to linear algebra, in preparation for the next two sections.  The algorithmic construction of a complete reduction for a given $\Sigma^*$-monomial (a generator of a $\Sigma^*$-extension modeling an indefinite sum) is derived in Section \ref{SECT:construct}. Applying these ideas iteratively (recursively) yields an algorithmic approach to construct a complete reduction in a tower of $\Sigma^*$-extensions in Section \ref{SECT:towers}. We also compare our method with the built-in algorithm of the Mathematica package {\tt Sigma}, show that this new framework can be used to solve the parameterized telescoping problem and relate complete reductions to reduced $\Sigma^*$-extensions which play important roles in Karr's structural theorem. A conclusion is given in Section~\ref{Sec:Conclusion}.

\section{Preliminaries}\label{SECT:pre}

Let $\bZ^{-}$ and $\bN$ denote the sets of negative and nonnegative integers, respectively. For an Abelian group $(G,\,+,\,0)$, we set $G^\times:=G \setminus \{0\}$. Throughout the paper, all fields are of characteristic zero. 
Let $F$ be a field and $t$ be an indeterminate over $F$. For  $p \in F[t]$, denote its leading coefficient and degree by $\lc_t(p)$ and $\deg_t(p)$, respectively. Note that $\lc_t(0):=0$ and $\deg_t(0):=-\infty$. We say that $p$ is monic with respect to $t$ if $\lc_t(p)=1$.  
For a rational function $f$ in $F(t)$, we call $a/b$ the {\em reduced representation} of $f$ if $a,b \in F[t]$ with $\gcd(a,\,b)=1$ and $b$ is monic with respect to $t$. Moreover, $a$ and $b$ are denoted by $\num(f)$ and $\den(f)$, respectively. If $\deg_t(\num(f))<\deg_t(\den(f))$, we say that $f$ is {\em $t$-proper}. In particular, $0$ is $t$-proper. The set consisting of $t$-proper elements of $F(t)$ is a linear space over $F$,  which is denoted by $F(t)_{(r)}$. 
Using polynomial division, each element $f\in F(t)$ can be uniquely written as the sum of a polynomial in $t$ and a $t$-proper element denoted by $\polypart(f)$ and $\proppart(f)$, respectively. Furthermore, we get the direct sum 
\begin{equation}\label{Equ:PolyFracSum}
F(t)=F[t] \oplus F(t)_{(r)}
\end{equation}
of $F$-vector spaces.

If it turns to algorithms, we assume that a field $F$ is \emph{computable}. This means that the standard field operations in $F$ are \emph{computable} (and thus one can apply the extended Euclidean algorithm in $F[t]$ and solve linear systems in $F(t)$). Furthermore, we assume that one can factorize polynomials in $E[x]$ over any multivariate rational function field extension $E=F(x_1,\dots,x_r)$ of $F$. Note that with our definition a rational function field $E$ over $F$ is computable if $F$ is computable.

Let $\sigma: F\to F$ be a field automorphism. The pair $(F,\,\sigma)$ is called a {\em difference field}.  We call $c \in F$ a {\em constant} if $\sigma(c)=c$. The set of constants $C_F=\{c\in F\mid \sigma(c)=c\}$ forms a subfield of $F$, which is called the \emph{constant field} of $(F,\sigma)$. The forward shift operator is defined by   

\begin{equation} 
 \begin{array}{cccc}
\Delta: & F & \longrightarrow & F \\
      &  f   & \mapsto         &  \sigma(f)-f
\end{array}
\end{equation}
which forms a $C_F$-linear map and satisfies for all $f_1,f_2 \in F$ the rule
\begin{equation}\label{Equ:Leibnitzrule}
\Delta(f_1f_2)=f_1\Delta(f_2)+\Delta(f_1)f_2+\Delta(f_1)\Delta(f_2).
\end{equation}
For a $C_F$-linear space $S\subset F$ we define the {\em summable space} of $S$ by
$$\Delta(S):=\{\sigma(f)-f\mid f\in S\}$$
which is again a $C_F$-linear space. 

 We call a difference field $(F,\, \sigma)$ \emph{computable} if the field $F$ is computable and the function $\sigma:F\to F$ is computable (i.e., can be executed by an algorithm).

Let $(F,\,\sigma)$ and $(E,\, \tilde{\sigma})$ be two difference fields. If $F$ is a subfield of $E$ and $\tilde{\sigma}|_{F}=\sigma$, then $(E,\, \tilde{\sigma})$ is called a {\em difference field extension} of $(F,\,\sigma)$. To introduce $\Sigma^*$-extensions, let $F(t)$ be the rational function field extension of $F$, i.e., $t$ is transcendental over $F$. Then for a given $a\in F$, there is a unique difference field extension $(F(t),\sigma)$ of $(F,\sigma)$ defined by $\sigma(t)=t+a$.

\begin{definition}
$(F(t),\sigma)$ as given above is called a {\em$\Sigma^*$-extension} of $(F,\sigma)$ if the set of constants remain unchanged, i.e, $C_{F(t)}=C_F$. Then $t$ is also called a {$\Sigma^*$-monomial} over $(F,\sigma)$ (or in short over $F$).
\end{definition}

We heavily rely on the following result~\cite{Karr1981}; for a proof see, e.g.,~\cite{Schn2016a}.

\begin{theorem}\label{Thm:SigmaKarr}
Let $(F(t),\sigma)$ be a difference field extension of $(F,\sigma)$ as given above. Then $t$ is a $\Sigma^*$-monomial over $(F,\sigma)$, i.e., $C_{F(t)}=C_F$ if and only if $\Delta(t)\notin\Delta(F)$.
\end{theorem}

In the following, we recall some notions (and properties) for the polynomials $p,q\in F[t]$ where $t$ is a $\Sigma^*$-monomial over $(F,\sigma)$.
\begin{itemize}
\item If $\gcd(p, \sigma^{\ell}(p))=1$ for any nonzero integer $\ell$, then it is said to be {\em $\sigma$-normal}; compare~\cite[Definition 2.2]{CDGHL2025}. Note that each monic irreducible polynomial with positive degree in $t$ is $\sigma$-normal by~\cite{Karr1981}; for a proof see~\cite[Cor.~1]{Bron2000} or \cite[Lemma~2.2.4]{Schn2001}. 
\item $f=\frac{p}{q}$ in reduced representation is said to be {\em $\sigma$-simple} if $f$ is $t$-proper and $q$ is $\sigma$-normal. 
\item We say that $p, q\in F[t]$ are {\em $\sigma$-coprime} if $\gcd(p, \sigma^{\ell}(q))=1$ for any nonzero integer~$\ell$.  
Note that the product of two 	
$\sigma$-normal and $\sigma$-coprime polynomials is $\sigma$-normal.
\item $p$ and $q$ are said to be {\em $\sigma$-equivalent} if $p=c\sigma^{\ell}(q)$ for some  $c \in F^{\times}$ and $\ell \in \bZ$. If this is the case, we also write $p\stackrel{\sigma}{\sim} q$. Note that $\stackrel{\sigma}{\sim}$ forms an equivalence relation. Further note that if $p$ and $q$ are $\sigma$-equivalent and $p,q\not\in F$ then $\ell\in\bZ$ and $c \in F^{\times}$ are unique: suppose we find different $\ell_1$ and $\ell_2$ with $c_1,c_2\in F^{\times}$ such that $p=c_1\sigma^{\ell_1}(q)$ and $p=c_2\sigma^{\ell_2}(q)$. Then $\sigma^{\ell_1-\ell_2}(q)=\sigma^{-\ell_2}(c_2/c_1) q$ which is not possible by~\cite[Cor.~1]{Bron2000} (i.e., there does not exist an element $f\in F[t]\setminus F$ with $f\mid \sigma^{\ell}(f)$ for some $\ell\in\bZ^{\times}$). 
\end{itemize} 

Finally, we recall effective bases introduced in \cite{DGLL2025}. They will help us to construct a complementary space of the given linear space in a dual manner in Sections \ref{SECT:construct} and \ref{SECT:towers}.

Let $K$ be a field with a subfield $E$, and $\Theta$ be an $E$-basis of $K$. 
For every $\theta  \in \Theta  $, we define the $E$-linear map $\theta^*: K \rightarrow E$ given the linear extension of
\[\theta^*(f)  =
\left\{
\begin{array}{ll}
1, &  \text{if $f =\theta$}, \\ \\
0,  & \text{if $f \in \Theta$ and $f \neq \theta$}.
\end{array} \right.
 \]
Let $\theta \in \Theta$ and $a \in K^{\times}$.
 We call $\theta$ {\em effective} for $a$ if $\theta^*(a) \neq 0$. 

\begin{definition}[\cite{DGLL2025}]\label{DEF:effective0}
The basis $\Theta$ is said to be effective  if there are two algorithms:
\begin{itemize}
\item [(i)]~given $a\in K^{\times}$, find $\theta \in \Theta$ with $\theta^*(a)\neq0$ (i.e., $\theta$ is effective for $a$); and 
\item [(ii)]~given $\theta \in \Theta  $ and $a \in K^{\times}$, compute $\theta^*(a)$.
\end{itemize}
\end{definition}

Let $K=F(t)$ be the field of rational functions in $t$ over the field $F$.
Set $T = \left\{ t^i \mid i \in \bN \right\}$ and
$M_{t}$ to be the set consisting of all monic and irreducible polynomials with positive degrees.
By the irreducible partial fraction decomposition, we see that
\begin{equation} \label{EQ:basis0}
  \Theta =  T \cup \left\{ \frac{t^i}{q^j} \mid q \in M_{t}, 0 \le i < \deg_t(q), j \in \bZ^+ \right\}
\end{equation}
is an $F$-basis of $F(t)$, which is called the {\em canonical $F$-basis} of $F(t)$. In our concrete case $\Theta$ is effective
with our assumption that $F$ is computable.
 More precisely, for the first task, we will use a variant of Algorithm {\tt BasisElement} from~\cite{DGLL2025} which takes care that $t$ appears in $\theta \in \Theta$ if $t$ appears in $a\in F(t)$. 

\medskip \noindent
{\tt Algorithm BasisElementForSummation}

\smallskip \noindent
{\tt Input:} $a\in F(t)^{\times}$ 

\smallskip \noindent
{\tt Output:} $(\theta, c) \in  \Theta \times F^\times$ with $c = \theta^*(a)$

\begin{enumerate}

	\item[(1)] $p \leftarrow \polypart(a),$ $r \leftarrow \proppart(a)$, $d \leftarrow$ $\den(r)$
    \item[(2)] {\tt if} $p \neq 0$ and $\deg_t(p)>0$ {\tt then} {\tt return} $\left(t^{\deg_t(p)}, \, \lc_t(p) \right)$ {\tt end if} 
    \item[(3)]  {\tt if} $r \neq 0$ {\tt then}
    \begin{itemize}
    
    \item[] $q \leftarrow$ a factor of $d$ in $M_t$, $m \leftarrow$ the multiplicity of $q$ in $d$
    \item[]  $h \leftarrow$ the coefficient of $q^{-m}$ in the $q$-adic expansion of $r$
    \item[]  {\tt return} $\left( t^{\deg_t(h)}/q^m,  \, \lc_t(h) \right)$
    \end{itemize}
    \item[] {\tt end if}
    \item[(4)]\label{BES:Step:Trivial} {\tt return} $(1, p)$
\end{enumerate}

We note that Algorithm  {\tt BasisElementForSummation} is nondeterministic and choosing different $q$'s in $d$ lead to different outputs. In practice, we choose $q$ to be the first member in the list of irreducible factors of $d$ computed by a factorization algorithm.

For the second task in Definition~\ref{DEF:effective0} we can use Algorithm~2.6 introduced in~\cite{DGLL2025}; for completeness we repeat it here.

\medskip \noindent
{\tt Algorithm Coefficient}

\smallskip \noindent
{\tt Input:} $(b, \theta) \in F(t) \times \Theta$

\smallskip \noindent
{\tt Output:} $\theta^*(b)$

\begin{enumerate}
	\item[(1)] $p \leftarrow \polypart(b),$ $r \leftarrow \proppart(b)$
	\item[(2)] Write $\theta = t^k/q^m$ for some $k, m \in \bN$, $q \in M_t$, $\gcd(t, q)=1$
	\item[(3)] {\tt if} $m = 0$ {\tt then}
	{\tt return} the coefficient of $t^k$ in $p$
	{\tt end if}
	\item[(4)] $h \leftarrow$ the coefficient of $q^{-m}$ in the $q$-adic expansion of $r$
	\item[(5)] {\tt return} the  coefficient of $t^k$ in $h$
\end{enumerate} 

Both algorithms will be used in Section \ref{SECT:towers} below.

\begin{example}\label{EX:BasisCoeff}
Let $E=\bQ$ and
$$f=\frac{3-x^2}{x^2+3x+2} \in E(x).$$
By the irreducible partial fraction decomposition, we have
$$f= -1+\frac{2}{x+1}+\frac{1}{x+2}.$$
Note that $\polypart(f)\in E$ and $\proppart(f) \neq 0$. So both $1/(x+1)$ and $1/(x+2)$ can be chosen as basis elements of $f$. If we take $\theta=1/(x+1)$, then $\theta^*(f)=2$.
Otherwise, if we take $f=1/(x+2)$, $\theta^*(f)$ is equal to $1$. So the function \texttt{BasisElementForSummation}($f$) may return $(\frac{1}{x+1}, 2)$ or $(\frac{1}{x+2}, 1)$. In particular, we have $\texttt{Coefficient}(f,\frac{1}{x+1})=2$, $\texttt{Coefficient}(f,\frac{1}{x+2})=1$. An other example is $\texttt{Coefficient}(f,x)=0$ or $\texttt{Coefficient}(f,1)=-1$.
\end{example}

Let $F_n=F_0(t_1,t_2,\cdots,t_n)$ be a field of multivariate rational functions with respect to $t_1,\cdots,t_n$ and let $C$ be a subfield of $F_0$ (later $C$ will take over the role of the constant field). For each $i$ with $0 \le i \le e$, set $F_i=F_{i-1}(t_i)$. Let $\Theta_i$ be the $F_{i-1}$-canonical basis of $F_i$ (given in~\eqref{EQ:basis0} with $F=F_{i-1}$ and $t=t_i$) and let $\Theta_0$ be a $C$-basis of $F_0$. It follows that
\begin{equation}\label{Equ:FullCThetaBasis}
\Xi=\{ \theta_0\theta_1\cdots \theta_n\mid \theta_i \in \Theta_i \}
\end{equation}
is a $C$-basis of $F_n$, which is called the {\em $\Theta_0$-canonical basis} of $F_n$. Here we define for $\theta=\theta_0\theta_1\cdots \theta_n\in \Xi$ the $C$-linear map $\theta^*:F_n\to C$ by $\theta^*=\theta^*_{0}\circ\dots\circ\theta^*_{n-1}\circ\theta_n^*$.
For $f\in F_n$, we define the {\em indicator} of $f$ to be 
$$\ind_n(f)=\min\{0 \le i \le n \mid f\in F_i\}.$$ 
Note that any $f\in F_n$ can be written in the form $f=\sum_{\xi_i \in \Xi} e_i \xi_i$ with $e_i\in C$ in the $\Theta_0$-canonical basis $\Xi$. As a consequence, 
\begin{equation}\label{Equ:IndToBasisEl}
\ind_n(f)=\max\{\ind_n(\xi_i) \mid e_i \neq 0\}.
\end{equation}

\begin{lemma}\label{Lemma:EffectiveBasisLifting}
Let $F_n$ be the rational function field over $F_0$ as introduced above, let $\Theta_0$ be a $C$-basis of $F_0$ and take the $\Theta_0$-canonical basis $\Xi$ of $F_n$ 
given in~\eqref{Equ:FullCThetaBasis}.
Then for any $f\in F_n^{\times}$ there is an effective $\theta\in \Xi$ for $f$ with $\ind_n(f)=\ind_n(\theta)$. In particular, if $F_0$ is computable and $\Theta_0$ is effective, $\Xi$ is effective and one can compute such a $\theta\in\Xi$ for $f$ together with $c=\theta^*(f)\in C$.
\end{lemma}
\begin{proof}
Let $f\in F_n^{\times}$ and define $k=\ind_n(f)$. We show the lemma by induction (recursion).
If $k=0$ then $f\in F_0^{\times}$ and we can write $f$ as a linear combination in the basis~$\Theta_0$. Then take any $\theta\in F_0$ that appears in this linear combination with the coefficient $c\in C^{\times}$. Note that $\theta \in F_0$ and thus $\ind_n(f)=0=\ind_n(\theta)$. If the basis $\Theta_0$ is effective, such a $\theta$ and $c$ can be computed which completes the base case. 
If $k\geq 1$ it follows by the constructions in Algorithm {\tt BasisElementForSummation} applied to $f$ and $F_{k-1}(t_k)$ (without claiming that one can compute the steps explicitly) that there are $\theta\in\Theta$ given in~\eqref{EQ:basis0} (with $F=F_k$ and $t=t_k$) and $c\in F_{k-1}$ with $\theta^*(f)=c\in F_{k-1}^{\times}$. Note that Algorithm {\tt BasisElementForSummation} can be executed if $F_{0}$ (and thus $F_{k-1}$) is computable.
By construction $c\in F_{k-1}$ and in the $F_{k-1}$-basis $\Theta_k$ the element $\theta\in\Theta_k$ occurs with the coefficient $c\in F_{k-1}^{\times}$. So by the induction assumption there exists  an effective $\theta'\in\Xi$ of $c$ with $\theta'^*(c)\in C^{\times}$; in particular, they can be computed if $F_0$ is computable.
Thus in the representation of $f$ in the $C$-basis $\Xi$ the coefficient of the basis element $\tilde{\theta}=\theta'\,\theta$ is $c'$, i.e., $\tilde{\theta}^*(f)=\theta'^*(\theta^*(f))=\theta'^*(c)=c'$. By construction $\ind_n(\tilde{\theta})=\ind_n(\theta'\theta)=\ind_n(\theta)=k$, which proves the lemma.
\end{proof}

\section{Complete Reductions in a $\Sigma^*$-extension}\label{SECT:construct}

We start with the following definition.

\begin{definition}
Let $(F,\sigma)$ be a difference field with constant field $C$. We call $(g, r)\in F^2$ a  {\em $\Sigma$-pair} of $f\in F$ if $f=\Delta(g)+r$.
Let $V$ be a $C$-subspace of $F$ such that $\Delta(V) \subset V$ and let $\phi: V \rightarrow V$ be a complete reduction for $\Delta(V)$ on $V$ (see Definition~\ref{Def:CompleteReductionVS}). In short, we will also say that $\phi$ is a {\em complete reduction of $\Delta(V)$ on $V$}. Furthermore, we say that $V$ is {\em $\phi$-computable}, if there is an algorithm which computes for given $f\in V$ a $\Sigma$-pair $(g, \phi(f))\in V^2$ of $f$, i.e.,  
$$f=\Delta(g)+\phi(f).$$
We call $\phi(f)$ the \emph{$\phi$-remainder} (or non-summable part w.r.t.\ $\phi$) of $f$ in $F$. If $V=F$ is $\phi$-computable, we also say that $(F,\sigma)$ is $\phi$-computable.
\end{definition}

Throughout this article we exploit the property that a complete reduction $\phi$ is idempotent: for all $f\in V$, $\phi(\phi(f))=\phi(f)$.
Before we present algorithms to compute a complete reduction for a $\Sigma^*$-monomial (and later for a tower of such extensions), we start with the following motivation in the context of additive decompositions. Suppose that we succeed in computing a $\Sigma$-pair $(g,r)$ for $f\in F$ in a given difference field $(F,\sigma)$ with $r=\phi(f)$ w.r.t.\ a complete reduction $\phi$ of $\Delta(F)$. Reinterpreting $f$, $f'$ and $g$ as sequences $f(k)$, $f'(k)$ and $g(k)$ in terms of summation objects, one gets~\eqref{Equ:RefinedTele} and summing this equation $k$ from $a$ to $b$ yields~\eqref{Equ:RefinedTeleSummed}.
Since $r=\phi(f)\notin\Delta(F)$, the sum on the left-hand side of~\eqref{Equ:RefinedTeleSummed} is not summable in $F$. 
This relation can be also connected to the following corollary.

\begin{corollary}\label{Cor:SigmaCheck}
Let $(F,\sigma)$ be a difference field and $\phi$ be a complete reduction of $\Delta(F)$ on~$F$. Let $(F(t),\sigma)$ be the difference field extension of $(F,\sigma)$ with $t$ transcendental over~$F$ and $\Delta(t)=f\in F^{\times}$. Then $t$ is a $\Sigma^*$-monomial over $(F,\sigma)$ if and only if $\phi(f)\neq0$. If this is the case, there exists also a $\Sigma^*$-monomial $s$ over $(F,\sigma)$ with $\Delta(s)=\phi(f)$.
\end{corollary}
\begin{proof}
Note that $\phi(f)\notin\Delta(F)$ and thus $\phi(F)\cap\Delta(F)=\{0\}$. Thus both statements follow by Theorem~\ref{Thm:SigmaKarr}. 
\end{proof}

Loosely speaking, the sum on the left-hand side \eqref{Equ:RefinedTeleSummed} establishes a $\Sigma^*$-monomial $t$ over $(F,\sigma)$ with $\Delta(t)=f$ if and only if the sum on the right-hand side is represented by a $\Sigma^*$-monomial $s$ over $(F,\sigma)$ with $\Delta(s)=r$.
Since $r$ is usually simpler than the given $f\in F^{\times}$, the $\Sigma^*$-monomial $s$ turns out to be more suitable for the design of an appropriate tower of $\Sigma^*$-extensions; these aspects will be explored further in Section~\ref{SUBSECT:wgt} below.

\medskip

In the following, we let $(F,\, \sigma)$ be a difference field with constant field $C$ and $t$ be a $\Sigma^*$ monomial over $F$ with $\Delta(t)=a\in F$. In addition, we assume that \begin{enumerate}
\item $(F,\sigma)$ is computable and $\phi$-computable for an explicitly given complete reduction $\phi: F \rightarrow F$ for $\Delta(F)$;
\item $F$ has an effective $C$-basis $\Theta$;
\item one can solve the shift-problem in $F[t]$; see Problem~SE below.
\end{enumerate}
With these algorithmic assumptions we will show in Theorem~\ref{TH:decomp} below that $(F(t),\sigma)$ is $\phi'$-computable for an explicitly given complete reduction $\phi':F(t)\to F(t)$ for~$\Delta(F(t))$. 
More precisely, we will decompose $F(t)$ to the direct sum~\eqref{Equ:PolyFracSum} with the proper rational part $F(t)_{(r)}$ and the polynomial part $F[t]$ utilizing the following property.

\begin{lemma}\label{LM:closure}
$\Delta(F[t])$ and $\Delta(F(t)_{(r)})$ are $C$-subspaces of $F[t]$ and $F(t)_{(r)}$, respectively.
\end{lemma}
\begin{proof}
 Since $\sigma(t)=t+a$ with $a\in F$, $\sigma$ preserves the degree of every polynomial in~$t$, we have
$\Delta(F[t])\subset F[t]~\text{and}~\Delta(F(t)_{(r)}) \subset F(t)_{(r)}$. Since $\Delta(F[t])$ and $\Delta(F(t)_{(r)})$ are closed under addition and multiplication with elements of $C$, they form subspaces of $F[t]$ and $F(t)_{(r)}$, respectively.
\end{proof}

We will construct complete reductions for $\Delta(F(t)_{(r)})$ on $F(t)_{(r)}$ and for $\Delta(F[t])$ on $\Delta(F[t])$ in Sections~\ref{SUBSECT:PR} and~\ref{SUBSECT:polyre}, respectively. Combining them, a complete reduction for $\Delta(F(t))$ on $F(t)$ is given in Section~\ref{Sec:CRForRatFul}.

\subsection{A complete reduction for proper rational functions}\label{SUBSECT:PR}

In this subsection, we generalize the additive decomposition~\cite{Abra1975} of proper rational functions in $C(x)$ with $\sigma(x)=x+1$ to general $\Sigma^*$-extensions. Here we will streamline and adapt the additive decomposition given in~\cite{Schn2007} to obtain a complete reduction.

We start with the following observation; the result follows immediately by~\cite[Cor.~4]{Bron2000}; see also~\cite[Lemma~5.1]{Schn2007}.

\begin{lemma}\label{LM:proppart}
	Let $h \in F(t)_{(r)}$ be $\sigma$-simple and assume that $h\in \Delta(F(t))$. Then $h=0$.
\end{lemma}
%
%

The above lemma suggests to reduce proper rational functions to alternative representations where the denominator is $\sigma$-normal. Here we will use the following lemma; compare~\cite[Theorem 5.1(i)]{Schn2007}.

\begin{lemma}\label{LM:reduce1}
Let $v, p \in F[t]$ with $\deg_t(v)<\deg_t(p)$ and $\ell \in \bZ$. Then there exists a $g\in F(t)_{(r)}$ such that
$$\frac{v}{\sigma^{\ell}(p)}=\Delta(g)+\frac{\sigma^{-\ell}(v)}{p}.$$
\end{lemma}

\begin{proof}
If $\ell\geq0$, we take $g=\sum_{i=1}^{\ell}\frac {\sigma^{-i}(v)}{\sigma^{\ell-i}(p)}$. Otherwise, if $l<0$, we take  $g$  to be $\sum_{i=0}^{-\ell-1} -\frac{\sigma^i(v)}{\sigma^{\ell+i}(p)}$. The desired result follows by a simple check. 
 \end{proof}

Note that two monic irreducible polynomials $p,q\in F[t]$ are $\sigma$-equivalent (i.e., $p\stackrel{\sigma}{\sim} q$) if there is a unique $k\in\bZ$ such that $\sigma^k(p)=q$. To turn our constructions below to algorithmic versions, we assume that this $k$, in case it exists, can be computed.

\medskip

\noindent{\tt Problem SE~(Shift Equivalence) in $F[t]$}. Given a $\Sigma^*$-monomial $t$ over $(F,\sigma)$ and irreducible monic polynomials $p,q\in F[t]$. Decide if $p$ and $q$ are shift-equivalent. If yes, return the unique $k\in\bZ$ such that $\sigma^k(p)=q$.

\medskip

Using this $\sigma$-equivalence, one can now refine a complete factorization 
$$f=f_1^{n_1}f_2^{n_2}\dots f_r^{n_r}$$
with monic irreducible pairwise coprime polynomials $f_i\in F[t]$ and $n_i\in\bN^{\times}$. Namely, one may pick among $f_1,\dots,f_r$ a subset $T=\{p_1,\dots,p_s\}$ which are pairwise $\sigma$-coprime and where all other irreducible factors in $f$ are $\sigma$-equivalent to $T$. This will lead to Karr's $\sigma$-factorization~\cite[Definition 23]{Karr1981}:
\begin{equation}\label{Equ:SigmaFacFull}
f=\prod_{j=u_1}^{v_1}\sigma^{j}(p_1)^{m_{1,j}}\dots\prod_{j=u_s}^{v_s}\sigma^{j}(p_s)^{m_{s,j}}
\end{equation}
with $m_{i,j}\in\bN$. Note that such a $\sigma$-factorization can be computed explicitly if $F$ is computable and one can solve Problem~SE in $F[t]$. We further remark that this construction is closely related to the greatest factorial factorization introduced in~\cite{Paul1995}.

For convenience, we will introduce the following notation. 
Let $p$ be a nonzero polynomial in $F[t]$ and $\alpha=\sum_{i=u}^{v} m_i \sigma^i \in \bN[\sigma,\sigma^{-1}]$, we write 
$$p^\alpha :=\prod_{i=u}^{v} \sigma^{i}(p)^{m_i}$$
where $\alpha$ is called the {\em $\sigma$-exponent} of $p$. Then we can write~\eqref{Equ:SigmaFacFull} in the compact form
\begin{equation}\label{EQ:SigmaFactorization}
	f=p_1^{\alpha_1} p_2^{\alpha_2}\cdots p_s^{\alpha_s},
\end{equation}
where $\alpha_i =\sum_{j=u_i}^{v_i} m_{i,j}\sigma^j \in \bN[\sigma,\,\sigma^{-1}]$.

We remark that the $\sigma$-factorization is unique up to $\sigma$-equivalent factors in $\{p_1,\dots,p_s\}$. E.g., taking $\{q_1,\dots,q_r\}$ with $q_i=\sigma^{-l_i}(p_i)$ with $l_i\in\bZ$, one gets
$$f=q_1^{\beta_1} q_2^{\beta_2}\cdots q_s^{\beta_s}$$
with $\beta_i=\sigma^{l_i}\cdot\alpha_i:=\sum_{j=u_i}^{v_i} m_{i,j}\sigma^{j+l_i} \in \bN[\sigma,\,\sigma^{-1}]$.

In our construction below it will be essential that also the elements for $T$ are fixed. To accomplish this task, we 
let $M_t$ to be the set consisting of all monic irreducible polynomials with positive degrees in $F[t]$ and fix a set $S \subset M_t$ of representatives of the equivalence classes induced by $\stackrel{\sigma}{\sim}$.
It follows that elements in $S$ are pairwise $\sigma$-coprime and for any $q\in M_t$ there is a unique $k\in\bZ$ and $p\in S$ such that $\sigma^{k}(p)=q$.
In particular, any $\sigma$-factorization of $f\in F[t]$ can be given in the form~\eqref{EQ:SigmaFactorization} where $p_1,\dots,p_r\in S$.

With this preparation we define the set of quotients whose denominators are built by factors of $S$:
\begin{equation}\label{EQ:ProperCompl}
	U_S:=\left\{\frac{a}{p_1^{m_1}\cdots p_{\ell}^{m_{\ell}}}\in F(t)_{(r)}\mid p_i \in S~\text{and}~m_i\in \bN  \right\}.
\end{equation}
Note that each element in $U_S$ is $\sigma$-simple, i.e., the denominators are $\sigma$-normal. Furthermore, we observe that $U_S$ is a $C$-linear space because any finite product of  $\sigma$-normal polynomials that are mutually $\sigma$-coprime is still $\sigma$-normal. We can now state the following result; compare~\cite[Thm.~5.1]{Schn2007}.

\begin{proposition}\label{Prop:decomposeProperpart}
We have
$$F(t)_{(r)}=\Delta(F(t)_{(r)})\oplus U_S.$$
In particular, for any $f\in F(t)_{(r)}$, there exist $g\in F(t)_{(r)}$ and $h \in U_S$ with $\deg_t(\den(h)) \le \deg_t(\den(f))$ such that  
\begin{equation}\label{EQ:Properreduction}
	f=\Delta(g)+h,
\end{equation}
and there is a complete reduction $\phi_S$ for $\Delta(F(t)_{(r)})$ on $F(t)_{(r)}$ with $\phi_S(f)=h$.\\
If $(F,\sigma)$ is computable and one can solve Problem~SE in $F[t]$, such $g$ and $h$ can be computed and $F(t)_{(r)}$ is $\phi_S$-computable. 
\end{proposition}

\begin{proof}
Let \eqref{EQ:SigmaFactorization} be a $\sigma$-factorization of the denominator of $f\in F(t)_{(r)}$ with $p_1,\dots,p_s\in S$; as remarked above, it can be computed if one can solve Problem~SE in $F[t]$.
Since $t$ is a $\Sigma^*$-monomial over $F$, all irreducible polynomials in $F[t]$ are $\sigma$-normal, which implies that $\sigma^{\ell_{i,j}}(p_i^{m_{i,j}})$ are pairwise coprime for all $i, j$. So by the extended Euclidean algorithm,
$$f=\sum_{i=1}^{s}\sum_{j=u_i}^{v_i} f_{i,j},$$
 where $f_{i,j}=v_{i,j}/\sigma^{\ell_{i,j}}(p_i^{m_{i,j}})$
for some $v_{i,j} \in F[t]$ with $\deg_t(v_{i,j})<m_{i,j}\deg_t(p_i)$.
Applying Lemma \ref{LM:reduce1} to each $f_{i,j}$, there is a $g_{i,j} \in F(t)_{(r)}$ such that 
$$f_{i,j}=\Delta\left(g_{i,j}\right)+\frac{\sigma^{-\ell_{i,j}}(v_{i,j})}{{p_i}^{m_{i,j}}}.$$
Summing up all the equations, we obtain
$f=\Delta(g)+h$
for some $g\in F(t)_{(r)}$ and $h\in U_S$ with $\deg_t(\den(h)) \le \deg_t(\den(f))$. 
Furthermore, it implies that $F(t)_{(r)}=\Delta(F(t)_{(r)})+U_S$. It follows from Lemma \ref{LM:proppart} that $\Delta(F((t)_{(r)})\cap U_S =\{0\}$ which establishes the direct sum $F(t)_{(r)}=\Delta(F(t)_{(r)})\oplus U_S.$
Summarizing, we get a complete reduction $\phi_S:F(t)_{(r)}\to F(t)_{(r)}$ for $\Delta(F(t)_{(r)})$ which maps $f$ to $h$. Furthermore, if $(F,\sigma)$ is computable and one can solve Problem~SE in $F[t]$, such $g$ and $h$ can be computed and thus $:F(t)_{(r)}$ is $\phi_S$-computable.
\end{proof}

We remark that for $(C(x),\sigma)$ with $\sigma(x)=x+1$ we have that $C_{C(x)}=C$, i.e., $x$ is a $\Sigma^*$-monomial over $C$. 
Since $\Delta(C[x])=C[x]$, Proposition~\ref{Prop:decomposeProperpart} leads to the following well known special case~\cite{Abra1975}.

\begin{corollary}\label{COR:ADRational}
	For the rational difference field $(C(x),\sigma)$ with $\sigma(x)=x+1$ we have 
	$$C(x)=\Delta(C(x)) \oplus U_S.$$
\end{corollary}

\begin{remark}\label{Remark:ExpandSetS}

In general, one may fix the infinite set $S$ for $F[t]$ a priori. More practically oriented, one may start with $S=\emptyset$ and adjoin stepwise new elements. More precisely, if an irreducible monic factor $q\in M_t$ arises in our constructions below, one checks if there is a $p\in S$ and $k\in\bZ$ with $\sigma^k(p)=q$ by solving Problem~SE. If this is not the case, one adjoins $p$ to $S$ (or any other factor $\sigma^l(q)$ with $\ell\in\bZ$) and proceeds. 
\end{remark}

The proof of Proposition~\ref{Prop:decomposeProperpart} and Remark~\ref{Remark:ExpandSetS} lead to the following algorithm.

\medskip \noindent
{\tt Algorithm ReductionForProperRationalFunctions}

\smallskip \noindent
{\tt Input:} a $\Sigma^*$-monomial $t$ over $(F,\sigma)$ in which one can solve Problem~SE; $f\in F(t)_{(r)}$. 

\smallskip\noindent From outside accessible: \textit{a finite set} $S\subseteq M_t$ whose elements are pairwise $\sigma$-coprime.

\smallskip \noindent
{\tt Output:} $(g,h) \in F(t)_{(r)} \times U_{S}$ such that \eqref{EQ:Properreduction} holds.

\smallskip \noindent
\begin{enumerate}
	\item[(1)] Compute a $\sigma$-factorization $p_1^{\alpha_1} \cdots p_s^{\alpha_s}$
	of $\den(f)$.
	\item[(2)]  {\tt for} $i$ {\tt from} $1$ {\tt to} $s$ {\tt do}
	
	\begin{itemize}
		\item[] Determine whether there exist $p \in S$ and $\ell\in\bZ$ such that $p_i=\sigma^\ell(p)$
		\item[] If yes, update $p_i$ to be $p$, and $\alpha_i$ to be $\sigma^{\ell}\cdot \alpha_i$, otherwise, append $p_i$ to $S$
	\end{itemize}
	{\tt end do}
	
	\item[(3)] Compute the partial fraction decomposition (using, e.g., the extended Euclidean algorithm)
	$$f=\sum_{i=1}^{s}\sum_{j=u_i}^{v_i} f_{i,j}$$
	where $f_{i,j}=v_{i,j}/\sigma^{\ell_{i,j}}(p_i^{m_{i,j}})$ with $\deg_t(v_{i,j})<m_{i,j}\deg_t(p_i)$
	\item[(4)] For each $f_{i,j}$, apply Lemma \ref{LM:reduce1} to find $g_{i,j} \in F(t)$ and $\tilde{v}_{i,j}\in F[t]$ such that
	$$f_{i,j}=\Delta(g_{i,j})+\frac{\tilde{v}_{i,j}}{p_i^{m_{i,j}}}$$
	
	\item[(5)] $g \leftarrow \sum_{i=1}^{s}\sum_{j=u_i}^{v_i}g_{i,j}$,\,   $h \leftarrow \sum_{i=1}^{s}\sum_{j=u_i}^{v_i} \tilde{v}_{i,j}/p_i^{m_{i,j}}$
	\item[(6)] {\tt return}  $(g,\, h)$ 
\end{enumerate}

\begin{example}\label{EX:pr}
	Let $F=\bQ(x)$ with the shift operator $\sigma(x)=x+1$ and $t$ be a $\Sigma^*$-monomial over $F$ with $\sigma(t)=t+1/(x+1)$ over $(F,\sigma)$; note that $t$ models the harmonic numbers $H_k=\sum_{i=1}^k\frac1{i}$ with $H_{k+1}=H_k+\frac1{k+1}$. Let
	$$f=-\frac{1}{(x+1)t^2+t} \in F(t)_{(r)}.$$
	Executing {\tt ReductionForProperRationalFunctions}$(F(t),f)$ with $S=\emptyset$ we proceed as follows. We get $\num(f)=-1/(x+1)$ and $\den(f)=t^2+t/(x+1)$. Further we compute a $\sigma$-factorization, say $t\sigma(t)$, of $\den(f)$. Thus we set $S=\{t\}$. Next, we find the partial fraction decomposition of $f$ as $1/t-1/\sigma(t)$. Applying Lemma \ref{LM:reduce1} to $1/\sigma(t)$ and $1/t$, we obtain $f=\Delta(1/t)$. Hence $f$ is summable in $F(t)_{(r)}$. In particular, we return $(g,h)=(1/t,0)$. 
\end{example}

\subsection{A complete reduction for polynomials}\label{SUBSECT:polyre}

In Section~\ref{SUBSECT:PR} we have established a complete reduction for the subspace $\Delta(F(t)_{(r)})$ summarized in Proposition~\ref{Prop:decomposeProperpart}.
In this subsection we will focus on constructing algorithmically  a $C$-linear subspace $V$ such that $F[t]=\Delta(F[t])\oplus V$.
This will yield the desired result
$$F(t)=\im(\Delta)\oplus (U_S+V).$$
Utilizing new techniques from \cite{DGLL2025}(see, e.g., Section~3 therein) for differential fields and extracting ideas from~\cite{Schn2007} (see, e.g., Lemma~4.3 therein) for difference fields we will proceed as follows:
\begin{enumerate}
\item Construct an auxiliary subspace $A$ such that $F[t] =\Delta(F[t])+ A$.
\item Compute a basis of $ \Delta(F[t]) \cap A$.
\item Fix a complementary subspace of $\Delta(F[t]) \cap A$ contained in $A$.
\end{enumerate}

For step~1 we define the $C$-subspace 
$$A=\bigoplus_{i\in\bN} \im(\phi) \cdot t^i$$
of $F[t]$ called the \emph{auxiliary subspace} for $\Delta(F[t])$ in $F[t]$.


The lemma given below reduces a polynomial in $F[t]$ to $A$ modulo $\Delta(F[t])$. In other words, each coefficient of a polynomial in $F[t]$ is reduced to a $\phi$-remainder.

\begin{proposition}\label{PROP:auxiliarySpace}
We have
$$F[t]=\Delta(F[t])+A.$$
In particular, for $p \in F[t]$ with $\deg_t(p)=d$, there exist $q \in F[t]$ with $\deg_t(q)\leq d$ and $r \in A$ with $\deg_t(r) \leq d$ such that
\begin{equation}\label{EQ:AuxPair}
 p=\Delta(q)+r.
 \end{equation}
If $(F,\sigma)$ is computable and $\phi$-computable, such $q$ and $r$ can be computed.
\end{proposition}

\begin{proof}
 We proceed by induction on $d$. If $d=0$, there are $q\in F$ and $r \in \im(\phi)$ such that $p=\Delta(q)+r$ by assumption. 
 
Assume that the proposition holds for $d-1$ with $d>0$. Set $p_d$ to be the leading coefficient of $p$. By the assumption again, there is a $\Sigma$-pair $(q_d,\,r_d)$ of $p_d$ with $r_d=\phi(p_d)$. It follows from~\eqref{Equ:Leibnitzrule} that
$$\Delta(q_dt^d)=\Delta(q_d)t^d+\Delta(q_d) \Delta(t^d)+q_d \Delta(t^d).$$
Note that $\Delta(t^d) \in F[t]$ is of degree less than $d$ and $\Delta(t)$ is a monic linear polynomial in $t$. So $\Delta(q_dt^d)=\Delta(q_d)t^d+ \text{lower terms}$. Thus 
\begin{align*}
p&=p_d\,t^d+\text{lower terms}\\
&=(\Delta(q_d)+r_d)t^d+\text{lower terms}\\
&=\Delta(q_d)t^d+r_dt^d+\text{lower terms}\\
&=\Delta(q_dt^d)+r_dt^d+\text{lower terms},
\end{align*}
i.e., $p=\Delta(q_dt^d)+r_dt^d+\tilde{p}$ for some $\tilde{p} \in F[t]$ with $\deg_t(\tilde{p})<d$. By the induction hypothesis, there are $\tilde{q} \in F[t]$ with $\deg_t(\tilde{q})<d$ and $\tilde{r} \in A$ with $\deg_t(\tilde{r}) <d$ such that $\tilde{p}=\Delta(\tilde{q})+\tilde{r}$. Consequently, the conclusion follows by setting $q=q_d\,t^d+\tilde{q}$ and $r=r_d\,t^d+\tilde{r}$, and observing that
$$\Delta(q)+r=\Delta(q_d\,t^d+\tilde{q})+r_d\,t^d+\tilde{r}=\Delta(q_d\,t^d)+r_d\,t^d+(\Delta(\tilde{q})+\tilde{r})=\Delta(q_d\,t^d)+r_d\,t^d+\tilde{p}=p.$$
In particular, if $(F,\sigma)$ is computable and $\phi$-computable, the above construction can be carried out explicitly.
\end{proof}

The underlying algorithm of Proposition \ref{PROP:auxiliarySpace} can be given as follows.

\medskip \noindent
{\tt Algorithm AuxiliaryReduction}

\smallskip \noindent
{\tt Input:} a $\Sigma^*$-monomial $t$ over $(F,\sigma)$ which is computable and $\phi$-computable, i.e., there is an algorithm \texttt{CompleteReduction} that computes $\Sigma$-pairs in $(F,\sigma)$ w.r.t.\ $\phi$; $p\in F[t]$

\smallskip \noindent
{\tt Output:} $(q,r)\in F[t] \times A$ such that \eqref{EQ:AuxPair} holds
\begin{enumerate}
\item[(1)] $\tilde{p} \leftarrow p$, $q \leftarrow 0$, $r \leftarrow 0$
\item[(2)]  {\tt while} $\tilde{p} \neq 0$ {\tt do}
\begin{itemize}
\item[] $(d,\,p_d) \leftarrow (\deg_t(\tilde{p}), \lc_t(\tilde{p}))$
\item[]  compute a $\Sigma$-pair $(q_d, \, r_d)$ of $p_d$ w.r.t. $\phi$ by executing\\
 Algorithm \texttt{CompleteReduction} in $(F,\sigma)$.
\item[] $(q,\,r) \leftarrow (q+q_dt^d, r+r_dt^d)$
\item[] $\tilde{p} \leftarrow \tilde{p}-\Delta(q_d\,t^d)-r_d\,t^d$
\end{itemize}
{\tt end do}
\item[(3)] {\tt return} $(q,\, r)$
\end{enumerate}

\begin{example}\label{EX:aux}
Let $F$ and $t$ be given in Example \ref{EX:pr}. By Corollary \ref{COR:ADRational}, there is a complete reduction $\phi$ for $\Delta(F)$ on $F$, and 
 $$\im(\phi)=\left\{\frac{a}{p_1^{m_1}\cdots p_{\ell}^{m_{\ell}}}\in \bQ(x)_{(r)}\mid p_i \in S_0~\text{and}~m_i\in \bN  \right\},$$
 where $S_0$ is a set of representatives of $\stackrel{\sigma}{\sim}$ in $\bQ[x]$. In particular, $(\bQ(x),\sigma)$ is $\phi$-computable.
 Here we suppose that $x\in S_0$, i.e., $x$ represents the class of all elements that are $\sigma$-equivalent to $x$.
We apply the algorithm \texttt{AuxiliaryReduction} to
$$p=- \frac{1}{x(1+x)}t^2+\frac{x^2+4x+1}{x(1+x)^2}t \in F[t]$$
and get
\begin{equation}\label{Equ:ExpAuxSplit}
q= \frac{t^2}{x}-\frac{1}{x^3}\in F[t] \quad \text{and} \quad r=\frac{t}{x}-\frac{1}{x^3}\in A 
\end{equation}
such that $p=\Delta(q)+r$.
\end{example}

We proceed with step~2 and are going to construct a $C$-basis of $\Delta(F(t)) \cap A$.

\begin{lemma}\label{LM:lc}
Let $0\neq p \in \Delta(F[t]) \cap A$. Then $\lc_t(p)=c\phi(\Delta(t))$ for some $c \in C^{\times}$.
\end{lemma}
\begin{proof}
Since $0\neq p \in \Delta(F[t])^{\times}$, there exists a $q \in F[t]$ such that $p=\Delta(q)$. Let $\deg_t(p)=d\geq0$. Then $\deg_t(q) \le d+1$ by \cite[Theorem 14]{Karr1981}. Write $q$ as $q_{d+1}t^{d+1}+\cdots+q_0$, where $q_i\in F$. By a straightforward computation, we have
$$p=\Delta(q_{d+1})t^{d+1}+(\Delta(t)(d+1)\sigma(q_{d+1})+\Delta(q_{d}))t^d+\text{lower terms}.$$
 Furthermore, $q_{d+1}=c$ for some $c\in C$ because $\deg_t(p)=d$. Note that $\lc_t(p)=c(d+1)\Delta(t)+\Delta(q_d)\neq0$. 
Since $p\in A$, $\lc_t(p)\in\im(\phi)$. Applying $\phi$ to both sides and using the fact that $\phi$ is idempotent it follows that
$$\lc_t(p)=\phi(\lc_t(p))=\phi(c(d+1)\Delta(t)+\Delta(q_d)) =c(d+1)\phi(\Delta(t))+\phi(\Delta(q_d))=\tilde{c}\phi(\Delta(t))$$
with $\tilde{c}=c(d+1)\in C$. Since $\lc_t(p)\neq0$, $\tilde{c}\in C^{\times}$. 
\end{proof}

As it turns out, $\phi(\Delta(t))$ arising in Lemma~\ref{LM:lc} will be crucial. More precisely, we need the following tuple.

\begin{definition}\label{Def:FirstPair}
A $\Sigma$-pair $(g_t, \phi(\Delta(t)))$ of $\Delta(t)$ in $F$ is called {\em a first pair associated to~$F(t)$}. 
\end{definition}

Note that such a first pair $(g_t, \phi(\Delta(t)))$ associated to $F(t)$ is not unique. More precisely, also $(g_t+c, \phi(\Delta(t)))$ with $c\in C$ is a first pair. For the second component the following property holds.

\begin{lemma}\label{Lemma:SecondEntryIsNonzero}
$\phi(\Delta(t))\in F^{\times}\cap\Delta(F[t])$. 
\end{lemma}
\begin{proof}
By Theorem~\ref{Thm:SigmaKarr}, $\Delta(t)\notin\Delta(F)$. Thus $\phi(\Delta(t))\in F^{\times}$.  Furthermore, $\Delta(t)=\phi(\Delta(t))+\Delta(g)$ for some $g\in F$. Thus $\phi(\Delta(t))=\Delta(t-g)$ and hence $\phi(\Delta(t))\in F^{\times}\cap\Delta(F[t])$.
\end{proof}

With the help of a first pair, a $C$-basis of the intersection can be constructed. For this task we start with the following lemma.

\begin{proposition}\label{PROP:intersection}
Let $(g_t, \phi(\Delta(t)))$ be a first pair associated to $F(t)$. Then
there exist
$$w_i=\frac{t^{i+1}}{i+1}-g_t t^i+\text{lower terms},\, i\in \bN$$
such that $B:=\{\Delta(w_i) \mid i\in \bN\}$ is a $C$-basis of $\Delta(F[t])\cap A$. Moreover, $\deg_t(\Delta(w_i))=i$ and the leading coefficient of $\Delta(w_i)$ is $\phi(\Delta(t))$ for each $i\in\bN$.
\end{proposition}
\begin{proof}
Recall that $\Delta(t)=\sigma(t)-t=a$. For $i\in\bN$, we have
\begin{align*}
	\Delta\left(\frac{t^{i+1}}{i+1}-g_tt^i \right) & = \frac{(t+a)^{i+1}-t^{i+1}}{i+1}-(\sigma(g_t)(t+a)^{i}-g_tt^i) \\
	& = \frac{(i+1)a t^i+\text{lower terms}}{i+1}-((\sigma(g_t)-g_t)t^i+\text{lower terms}) \\
	& =(\Delta(t)-\Delta(g_t))t^i+\tilde{w}_i \quad \text{for some $ \tilde{w}_i \in F[t]$ with $\deg_t(\tilde{w}_i)<i$}\\
	& =\phi(\Delta(t))t^i+\tilde{w}_i.
\end{align*}
It follows from Proposition~\ref{PROP:auxiliarySpace} that there exist a $q_i \in F[t]$ with $\deg_t(q_i)<i$ and an $r_i \in A$ with $\deg_t(r_i) <i$ such that $\tilde{w}_i=\Delta(q_i) +r_i$. So
\begin{equation}\label{EQ:basis}
\Delta\underbrace{\left(\frac{t^{i+1}}{i+1}-g_tt^i -q_i\right)}_{w_i}=\phi(\Delta(t))t^i+r_i.
\end{equation}
By Lemma~\ref{Lemma:SecondEntryIsNonzero}, $\phi(\Delta(t))t^i\neq0$. Thus $\deg(\Delta(w_i))=i$ and $\lc_t(\Delta(w_i))=\phi(\Delta(t))$.  
Since both $\phi(\Delta(t))t^i$ and $r_i$ belong to  $A$, we have $\Delta(w_i) \in A$. 
Thus $B=\{\Delta(w_i)\mid i\in\bN\}$ is contained in $I:=\Delta(F[t]) \cap A$. Moreover, $B$ is a $C$-linear independent set because $\Delta(w_i)$ is of degree $i$ for each $i\in \bN$.
Finally, we will show that $B$ spans $I$. Assume that $q \in I$ with $q \neq 0$. By Lemma \ref{LM:lc}, $\lc_t(q)=c\phi(\Delta(t))$ for some $c\in C^{\times}$. Let $d=\deg_t(q)$ and $\tilde{q}=q-c\Delta(w_d)$. Then $\tilde{q} \in I$ with $\deg_t(\tilde{q})<d$. By induction on $d$, we conclude that $q$ can be expressed as a $C$-linear combination of the elements in $B$.
\end{proof}

We call $B$ given in Proposition \ref{PROP:intersection} an {\em echelon basis} of $\Delta(F[t]) \cap A$ induced by a first-pair $(g_t, \phi(\Delta(t)))$. Furthermore, looking at the proof of Proposition~\ref{PROP:intersection} we obtain the following algorithm to compute such an echelon basis.
 
\medskip \noindent 
{\tt Algorithm EchelonBasis} 

\smallskip \noindent
\text{\tt Input:} $k\in \bN$; a $\Sigma^*$-monomial $t$ over $(F,\sigma)$ which is computable and $\phi$-computable;
a first pair  $(g_t, v)$ with $v=\phi(\Delta(t))$ associated to $F(t)$

\smallskip \noindent
\text{\tt Output:} a list $L=[(w_0,b_0),\cdots, (w_k, b_k)]$, where $b_i=\Delta(w_i)$ and $w_i$ is given in \eqref{EQ:basis} \

\begin{itemize}
  \item[(1)]  $L \leftarrow [(t - g_t, v)]$
  \item[(2)]~{\tt for} $i$ {\tt from} 1 {\tt to} $k$  {\tt do}
\begin{itemize}
 \item[] $a \leftarrow t^{i+1}/(i+1) - g_t t^i$, \,\, $ \tilde{w}\leftarrow \Delta(a)-v\,t^i$
\item[] $(q, r)\leftarrow$ {\tt AuxiliaryReduction}$(F,\tilde{w})$ 
\item[] $(w, b) \leftarrow (a - q, \, v t^i + r)$
\item[] $L \leftarrow$ the list obtained by appending $(w,b)$ to $L$
  \end{itemize}
  {\tt end do}

\item[(3)]~ {\tt return} $L$
\end{itemize}

\begin{example} \label{EX:basis}
Let $F(t)$ be the same as Example \ref{EX:pr}. A first pair associated to $F(t)$ is $(\frac{1}{x}, \frac{1}{x})$. Applying the algorithm with this first pair and $k=1$ returns the first 2 elements of the echelon basis of $\Delta(F[t]) \cap A$:
$$b_0=\frac{1}{x} \quad \text{and} \quad b_1=\frac{t}{x}-\frac{1}{2x^2}.$$
The above algorithm also computes the corresponding pre-images:
$$w_0=t-\frac{1}{x}\quad \text{and} \quad w_1=\frac{t^2}{2}-\frac{t}{x}+\frac{1}{2x^2}.$$
\end{example}

In the last step we define a complementary space of $\Delta(F[t])$ to represent elements in the desired decomposition. For this final step we introduce the following subspace of $A$. Here we suppose that $\Theta$ is a $C$-basis of $F$.

\begin{definition}\label{Def:ComplementSpaceForPoly}
For $\theta\in\Theta$ 
we define the {\em $\theta$-complement} of $\Delta(F[t])$ in $F[t]$ by
\begin{equation}\label{EQ:Polycompl}
	V_{\theta}=\bigoplus_{i\in \bN} (\im(\phi) \cap \ker(\theta^*)) t^i\subseteq A.
\end{equation}
\end{definition}

\begin{proposition}\label{PROP:complement}
Let $\theta \in \Theta$ be effective for $\phi(\Delta(t))$ and let $V_{\theta}$ be the $\theta$-complement of $\Delta(F[t])$ in $F[t]$. Then
$A =(\Delta(F[t]) \cap A) \oplus V_\theta$
and
\begin{equation}\label{Equ:DirectSumPolyPart}
F[t]=\Delta(F[t]) \oplus V_{\theta}.
\end{equation}
In particular, for $p\in F[t]$ with $\deg_t(p)=d$, there are $q\in F[t]$ and $v\in V_{\theta}$ with $\deg(v) \le d$ such that
\begin{equation}\label{Equ:pDelta(q)Rel}
p=\Delta(q)+v.
\end{equation}
Further, there is a complete-reduction $\phi_{\theta}$ for $\Delta(F[t])$ on $F[t]$ with $\phi_{\theta}(f)=v$.\\
If $(F,\sigma)$ is computable and $\phi$-computable, such $q$ and $v$ can be computed and $F[t]$ is $\phi_{\theta}$-computable. 
\end{proposition}
\begin{proof}
Let $(g_t, \, \phi(\Delta(t)))$ be a first pair associated to $F(t)$.
First we show that $A =(\Delta(F[t]) \cap A) + V_{\theta}$. Since
$V_{\theta}$ is contained in $A$, we have $(\Delta(F[t]) \cap A) +V_{\theta} \subset A$. Conversely, let $p \in A$ with $p\neq 0$. Write $p$ as $p_dt^d+\cdots+p_0$,
where $p_i \in F $ with $p_d \neq 0$. Let 
\begin{equation}\label{Equ:EchelonBasisForProof}
B=\{\Delta(w_i)\mid i\in \bN\}
\end{equation}
be the echelon basis induced by $(g_t, \, \phi(\Delta(t)))$ of
the intersection $\Delta(F[t])\cap A$ given in Proposition~\ref{PROP:intersection}. Set $c=\theta^*(p_d)/\theta^*(\phi(\Delta(t)))\in C$ and
$r=p-c\Delta\left(w_{d}\right).$  Then
$$r=\underbrace{(p_d-c\phi(\Delta(t))}_{=:r_d}t^d+\tilde{r}$$
with $\tilde{r} \in F[t]$ where $\deg_t(\tilde{r})<d$. Note that $p_d, \phi(\Delta(t)) \in \im(\phi)$. So $r_d \in \im(\phi)$. Furthermore
$$\theta^*(r_d)=\theta^*(p_d)-c\theta^*(\phi(\Delta(t)))=\theta^*(p_d)-\frac{\theta^*(p_d)}{\theta^*(\phi(\Delta(t)))}\theta^*(\phi(\Delta(t)))=0.$$
Thus $r_d\in \im(\phi)\cap \ker(\theta^{*})$. On the other hand, $\tilde{r} \in A$ because $p, \Delta\left(w_{d}\right), r_dt^d \in A$. 
The conclusion follows by induction on $d$. Namely, suppose that we get $v'\in V_{\theta}$ and $q'\in F[t]$ such that $\tilde{r}=v'+\Delta(q')$ where $\Delta(q')\in A$; for the base case $\tilde{r}=0$ this holds trivially.
Now define $v=r_d\,t^d+v'\in F[t]$ and $q=c\,w_d+q'\in F[t]$. Note that $\deg(v)\leq d$ ($r_d$ might be $0$) and 
$v\in V_{\theta}$ since $r_d\in\im(\phi)\cap \ker(\theta^{*})$ and $v'\in V_{\theta}$. Further note that $\Delta(q)\in A$ since $\Delta(q'),\Delta(w_d)\in A$. Thus with
\begin{align*}
v+\Delta(q)&=r_d\,t^d+v'+\Delta(c\,w_d+q')\\
&=c\,\Delta(w_d)+(r_d\,t^d+v'+\Delta(q'))\\
&=c\,\Delta(w_d)+(r_d\,t^d+\tilde{r})\\
&=c\,\Delta(w_d)+r=p
\end{align*}
we conclude that $p\in (\Delta(F[t]) \cap A) + V_{\theta}$. \\
Next, we prove $(\Delta(F[t])\cap A)\cap V_{\theta}=\{0\}.$ 
Suppose the contrary and take  $p \in (\Delta(F[t])\cap A)\cap V_{\theta}=\Delta(F[t])\cap V_{\theta}$ with $p \neq 0$. Then $p \in \Delta(F[t])\cap A $, which implies that $p=\sum_{i=0}^{d}c_i\Delta(w_i)$ for some $c_i \in C$ with $c_d\neq0$ by Proposition \ref{PROP:intersection}. 
Since $p \in V_{\theta}$, we have that $\lc_t(p)\in\ker(\theta^*)$ and thus $\theta^*(\lc_t(p))=0$. By Proposition~\ref{PROP:intersection} we have that $\deg_t(\Delta(w_d))=d$ and $\lc_t(\Delta(w_d))=\phi(\Delta(t))$. So $\lc_t(p)=c_d \lc_t(\Delta(w_d))=c_d\,\phi(\Delta(t))$. Thus $\theta^*(\lc_t(p))=c_d\,\theta^*(\phi(\Delta(t)))\neq0$ since $\theta$ is effective for $\phi(\Delta(t))$, a contradiction. Thus $(\Delta(F[t])\cap A)\cap V_{\theta}=\{0\}$ which implies that $A =(\Delta(F[t]) \cap A) \oplus V_{\theta}.$ \\
By Proposition \ref{PROP:auxiliarySpace} and $A =(\Delta(F[t]) \cap A) + V_{\theta}$, we have $F[t]=\Delta(F[t]) +V_{\theta}$. Furthermore, $\Delta(F[t]) \cap V_{\theta} \subset \Delta(F[t]) \cap V_{\theta} \cap A =\{0\}$ which implies the direct sum~\eqref{Equ:DirectSumPolyPart}. \\
Now let $p'\in F[t]$ with $\deg_t(p')=d$. Then by Proposition~\ref{PROP:auxiliarySpace} we get $p\in A$ with $\deg_t(p)\leq d$ and $q'\in F[t]$ such that $p'=\Delta(q')+p$. 
By the construction above we obtain $v\in  V_{\theta}$ with $\deg_t(v)\leq d$ and $q\in F[t]$ such that $p=\Delta(q)+v$. This gives $p'=\Delta(q')+p=\Delta(q')+\Delta(q)+v=\Delta(\tilde{q})+v$ with $\tilde{q}=q+q'\in F[t]$ which proves the statement in~\eqref{Equ:pDelta(q)Rel}. \\
If $(F,\sigma)$ is computable and $\phi$-computable, the first pair $(g_t,\phi(\Delta(t))$ and the first $d+1$ basis elements of the echelon basis $B$ induced by the first pair can be computed with Algorithm~\texttt{EchelonBasis}. Moreover, Algorithm \texttt{AuxiliaryReduction} is applicable to compute $p$ and $q'$. Hence also the constructions for $v$ and $q$ can be carried out. We refer to Algorithm~\texttt{ReductionForPolynomials} for a detailed summary.
\end{proof}

For the construction above one has to fix (among different possible choices) an effective $\theta\in\Theta$ to define the $\theta$-complement $V_{\theta}$. 
This motivates the following definition.

\begin{definition}
	Let $\theta \in \Theta$ be effective for $\phi(\Delta(t))$ and $c=\theta^*(\phi(\Delta(t)))\in C^{\times}$. Then $(\theta, c)$ is called a {\em second pair} associated to $F(t)$.
\end{definition}

We remark that such a pair can be computed if $\Theta$ is effective by executing Algorithm {\tt BasisElementForSummation}. 
We note further that the proof of Proposition~\ref{PROP:complement} is independent of the choice of the first-pair $(g_t+c, \, \phi(\Delta(t)))$ for some $c\in C$. More precisely, one might get another echelon basis~\eqref{Equ:EchelonBasisForProof} by using a different first pair. However, the result of $q$ and $v$ is invariant of this choice: for any other $\Sigma$-pair $(q',v')\in F[t]\times V_{\theta}$ of $p\in F[t]$ we have that $v=v'$ and $q-q'\in C$. Later the above construction will be applied several times for different inputs $p\in F[t]$. For efficiency reasons we will therefore also fix a first pair for a given $\Sigma^*$-monomial. In this way, we can reuse the already derived basis elements from Algorithm \texttt{EcholonBasis} when we run again in a problem to compute a complete reduction for $\Delta(F[t])$ on $F[t]$ (coming, e.g., from recursive calls). 

\medskip

Summarizing, we can extract from the proof of Proposition~\ref{PROP:complement} the following algorithm where we fix the first and second pairs associated to $F(t)$ accordingly.

\medskip \noindent
{\tt Algorithm ReductionForPolynomials}

\smallskip \noindent
\text{\tt Input:} a $\Sigma^*$-monomial $t$ over $(F,\sigma)$ which is computable and $\phi$-computable; $p\in F[t]$.

\smallskip\noindent From outside accessible: first and second pairs $(g_t, \, \phi(\Delta(t)))$, $(\theta,\,c)$ associated to $F(t)$.

\smallskip \noindent
\text{\tt Output:} $(q,\,v) \in F[t] \times V_{\theta}$ such that $p=\Delta(q)+v$

\begin{itemize}
	\item[(1)]~$(q, \, r) \leftarrow {\tt AuxiliaryReduction}(F(t),p)$
	\item[(2)]~  $v \leftarrow r$, $d \leftarrow \deg_t(r)$
	\item[(3)]~ $L \leftarrow {\tt EchelonBasis}(d,F(t),(g_t,\phi(\Delta(t))))$
	\item[(4)] ~{\tt for} $i$ {\tt from} $d+1$ {\tt to} $1$  {\tt do}
	\begin{itemize}
		\item[] $a \leftarrow$ the coefficient of $t^{i-1}$ in $v$,
		$\tilde{c} \leftarrow \theta^*(a)$
		
		\item[] $(w, b) \leftarrow L[i]$
		
		\item[] $q \leftarrow q + c^{-1}\tilde{c} w$, \, $v \leftarrow v - c^{-1}\tilde{c} b$\,
	\end{itemize}
	{\tt end do}
	
	\item[(5)]~ {\tt return} $(q, v)$
\end{itemize}

\begin{example}\label{EX:Proj}
For the $\Sigma^*$-monomial $t$ over $F$ and $r$ given in Example \ref{EX:pr}, we take the first-pair $(1/x,\, 1/x)$ associated to $F(t)$.  So $1/x$ is the only effective element in $\phi(\Delta(t))$. Thus as the second pair we can choose $(1/x,\, 1)$. Now we reduce $p$ from Example \ref{EX:aux} to the $\theta$-complement. 
We have already derived the decomposition $p=\Delta(q)+r$, where $q=t^2/x-1/x^3$ and $r=t/x-1/x^3 \in A$; compare~\eqref{Equ:ExpAuxSplit}. Since $\deg_t(r)=1$, we only need the first two elements of the echelon basis which have been derived in Example~\ref{EX:basis}. Finally we project $r$ to $V_{\theta}$. The result is $r=\Delta(w_1)+v$, where $w_1$ is given in Example~\ref{EX:basis} and 
$v=\frac{1}{2x^2}-\frac{1}{x^3} \in V_{\theta}.$
It follows that $p=\Delta(q+w_1)+v.$
\end{example}

\subsection{A complete reduction for rational functions}\label{Sec:CRForRatFul}

Combining Sections~\ref{SUBSECT:PR} and~\ref{SUBSECT:polyre}, we obtain a complete reduction for $\Delta(F(t))$.

\begin{theorem} \label{TH:decomp}
Let $t$ be a $\Sigma^*$-monomial over $(F,\sigma)$ and suppose that there is a complete reduction $\phi$ for $\Delta(F)$ on $F$. Let $\Theta$ be a $C$-basis of $F$ and $\theta \in \Theta$ be effective for $\phi(\Delta(t))$. Furthermore, let $S$ be a set of representatives of the equivalence classes induced by $\stackrel{\sigma}{\sim}$ on $M_t$.
Then there is a complete reduction $\psi_{(S,\theta)}$ for $\Delta(F(t))$ on $F(t)$ which establishes the direct sum
\begin{equation} \label{EQ:decomp}
  F(t) = \Delta(F(t)) \oplus U_S  \oplus V_{\theta},
\end{equation}
where $U_S$ and $V_{\theta}$ are given in \eqref{EQ:ProperCompl} and \eqref{EQ:Polycompl}, respectively.\\ 
If $(F,\sigma)$  is computable and $\phi$-computable and one can solve Problem~SE in $F[t]$, then $(F(t),\sigma)$ is $\psi_{(S,\theta)}$-computable. 
\end{theorem}
\begin{proof}
$F(t)$ can be written as the direct sum of $F[t]$ and $F(t)_{(r)}$. By Proposition~\ref{Prop:decomposeProperpart}, we have $F(t)_{(r)}=\Delta(F(t)_{(r)}) \oplus U_S$. On the other hand, Proposition \ref{PROP:complement} implies that $F[t]=\Delta(F[t]) \oplus V_{\theta}$. So $F(t)=\Delta(F[t]) \oplus \Delta(F(t)_{(r)}) \oplus U_S \oplus V_{\theta}$. Consequently,  $ F(t) = \Delta(F(t)) \oplus U_S  \oplus V_{\theta}$. Furthermore, take $f\in F(t)$ and let $f_1\in F(t)_{(r)}$ and $f_2\in F[t]$ with $f=f_1+f_2$. Then there exist $g\in F(t)_{(r)}$ and $h \in U_S$ such that $f_1=\Delta(g)+h$ by Proposition~\ref{Prop:decomposeProperpart}, and there exist $q\in F[t]$ and $v\in V_{\theta}$ such that $f_2=\Delta(q)+v$ by Proposition~\ref{PROP:complement}. This yields $f=f_1+f_2=h+v+\Delta(g')$ with $g'=g+q\in F(t)$ and defines the complete reduction $\psi_{(S,\theta)}$ for $\Delta(F(t))$ on $F(t)$ with $\psi_{(S,\theta)}(f)=h+v\in U_s\oplus V_{\theta}$. Finally, if
$(F,\sigma)$ is computable and $\phi$-computable and one can solve Problem~SE in $F[t]$, then $g'$, $h$ and $v$ can be computed explicitly and $(F(t),\sigma)$ is $\psi_{(S,\theta)}$-computable.
\end{proof}

The algorithmic version of the theorem can be summarized as follows.

\medskip \noindent 
{\tt Algorithm CompleteReduction} 

\smallskip \noindent
\text{\tt Input:} a $\Sigma^*$-monomial $t$ over $(F,\sigma)$ which is computable, $\phi$-computable and where Problem~SE is solvable in $F[t]$; $p\in F[t]$.

\smallskip\noindent From outside accessible: a first and second pair $(g_t, \, \phi(\Delta(t)))$  and $(\theta,\,c)$ associated to $F(t)$ and a finite set $S\subseteq M_t$ whose elements are $\sigma$-coprime.
\smallskip \noindent

\smallskip

\noindent\text{\tt Output:} $(g,\,r) \in K(t) \times (U_S\oplus V_{\theta})$ such that $f=\Delta(g)+r$

\begin{itemize}
  \item[(1)]\, $f_1 \leftarrow \proppart(f)$\, $f_2\leftarrow \polypart(f)$
  \item[(2)]\, $(g, h) \leftarrow$ {\tt ReductionForProperRationalFunctions}$(F(t),f_1)$
  \item[(3)]\, $(q, v) \leftarrow$ {\tt ReductionForPolynomials}$(F(t),f_2)$
  \item[(4)]\, {\tt return} $(g+q, h+v)$
\end{itemize}

\begin{example}\label{EX:cr}
We start with the sum
$$S(n)=\sum_{k=1}^{n}\frac{k(k^2+ 5k +4) H_k^3 +(k^2 +4 k+1) H_k^2- (k+1)^2 H_k^4-k - 2 k^2 - k^3}{ k (1 + k)^2 (1 + H_k + kH_k )H_k}$$
in terms of the harmonic numbers $H_k=\sum_{i=1}^k\frac{1}{k}$. Taking the difference field $(F,\sigma)$ and the $\Sigma^*$-monomial $t$ over $F$ as given in Example \ref{EX:pr} we can represent the summand by
$$ f=\frac{x(x^2+ 5x +4) t^3 +(x^2 +4 x+1) t^2- (x+1)^2 t^4-x - 2 x^2 - x^3}{ x (1 + x)^2 (1 + t + t x)t} \in F(t).$$
By polynomial division we get the splits
$$f_1=\proppart(f)=-\frac{1}{(1+x)t^2+t} \quad \text{and} \quad f_2=\polypart(f)= \frac{(x^2+4x+1)t-(1+x)t^2}{x(1+x)^2}.$$
By Example \ref{EX:pr}, we find $f_1=\Delta(1/t)$, which, together with Example \ref{EX:Proj}, implies that
\begin{equation}\label{Equ:RefinedTeleExp}
f=\Delta\left(\underbrace{\frac{2+x}{2x}t^2-\frac{t}{x}+\frac{x-2}{2x^3}+\frac{1}{t}}_{=g}\right)+\underbrace{\frac{x-2}{2x^3}}_{=r}
\end{equation}
where $g\in F(t)$ and $r\in V_{\theta}$. Reinterpreting the $\Sigma$-pair $(g,r)$ of $f$ to $(g(k),r(k))$ in terms of the harmonic numbers, equation~\eqref{Equ:RefinedTeleExp} can be restated as~\eqref{Equ:RefinedTele}, and summing this equation over $k$ from $1$ to $n$ leads to the identity~\eqref{Equ:RefinedTeleSummed} (with $a=1$, $b=n$) where the remainder sum
$$\sum_{k=1}^nr(k)=\sum_{k=1}^n\frac{k-2}{2k^3}=\frac{H_n^{(2)}}{2}-H_n^{(3)}$$
can be rewritten in terms of the generalized harmonic numbers $H_n^{(o)}=\sum_{i=1}^n\frac{1}{i^o}$. Summarizing, we obtain the simplification
$$S(n)=\tfrac{(n+1)^2(n+3)H_n^3+(n+1)(n+7)H_n^2-2(n^3+3n^2+3n-1)H_n+2n(n+1)^2}{2(n+1)^2(1+(n+1)H_n)} +\frac{H_n^{(2)}}{2}-H_n^{(3)}.$$
\end{example}

We note that $\psi_{(S,\theta)}|_{F}\neq\phi$. This observation is implied by the simple fact that
 $\Delta(t)\in F$ and $\Delta(t)\notin\Delta(F)$, i.e., $\phi(\Delta(t))\neq0$, but $\Delta(t)\in\Delta(F(t))$, i.e, $\psi_{(S,\theta)}(\Delta(t))=0$.
However $\phi$ and $\psi_{(S,\theta)}$ are closely related as carried out in the next corollary. The specialization to polynomials in statement (2) will be used in Corollary~\ref{Cor:SringExt} below.

\begin{corollary}\label{COR:remainder}
	Let $\phi: F\rightarrow F$ be a complete reduction for $\Delta(F)$, $(\theta, c)$ be a second pair associated to $F(t)$ and $\psi_{(S,\theta)}$ be the complete reduction given in Theorem \ref{TH:decomp}. Then
	\begin{enumerate}
		\item For every $f \in F$, we have that $\psi_{(S,\theta)}(f) = \phi(f) + \tilde{c} \phi(\Delta(t))$, where $\tilde{c}  = - \theta^*\left(\phi(f)\right) c^{-1}.$
		\item For every $f\in F[t]$ we have that $\psi_{(S,\theta)}(f)\in V_{\theta}$.
	\end{enumerate}
\end{corollary}
\begin{proof}
(1)	Let $f\in F$. Then there is a $g\in F$ such that $f=\Delta(g)+\phi(f)$. Set $v:=\phi(f) + \tilde{c} \phi(\Delta(t))$. Then  
	$f=\Delta(g)+\phi(f)=\Delta(g)+v-\tilde{c} \phi(\Delta(t))$. Since $\phi(\Delta(t)) \in \Delta(F[t])$ by Lemma~\ref{Lemma:SecondEntryIsNonzero}, there is a $g'\in F[t]$ with $f=v+\Delta(g')$. With
	$\tilde{c}=-\theta^*(\phi(f))c^{-1}=-\theta^*(\phi(f))/\theta^*(\phi(\Delta(t)))$
	we get
	$\theta^*(v)=\theta^*(\phi(f))+\tilde{c}\theta^*(\phi(\Delta(t)))=0$. Since $v \in A \cap F$, 
	$v$ belongs to the $\theta$-complement, which implies that $\psi_{(S,\theta)}(f)=v$ by Theorem \ref{TH:decomp}.\\
(2)	Let $f\in F[t]$. Then in the construction we get $f_1=\proppart(f)=0$ and $f_2=\polypart(f)=f$. Following the proof of Theorem~\ref{TH:decomp} we get $g\in F[t]$ and $v\in V_{\theta}$ such that $f=f_2=\Delta(g)+v$. In particular, $\psi_{S,\theta}(f)=v$. 
	\end{proof}

In~\cite{Schn2007} a similar construction has been provided to decompose $f\in F(t)$ in a summable and non-summable part where $t$ is a $\Sigma^*$-monomial (or a $\Pi$-monomial\cite{Karr1981} covering also products). 
While the construction for the rational part is only a streamlined version~\cite{Schn2007}, the construction of the polynomial part has been substantially improved in this article for the case of a $\Sigma^*$-monomial. First, our new construction provides a complete reduction and not only an additive decomposition (i.e., the complementary set forms a subspace of $F(t)$). Second, our algorithm
does not require to solve any underlying difference ring equation in $(F,\sigma)$ using classical tools such as degree/denominator bounds~\cite{Karr1981,Bron2000,Schn2001} and solving the underlying linear system. Finally, all the coefficients of the polynomial contributions are reduced and thus lead to significantly smaller representations. We conclude this section with the following optimality property of $\psi_{(S,\theta)}(f)$; compare~\cite{Schn2007} (see Corollaries~4.1 and~5.1 therein).


\begin{corollary}\label{COR:minimality}
	Let $f\in F(t)$ with the $\psi_{(S,\theta)}$-remainder $h$. Assume that there exist $\tilde{g},\tilde{h} \in F(t)$ such that $f=\Delta(\tilde{g})+\tilde{h}$. Then
	$$\deg_t(\den(\proppart(h)) \le \deg_t(\den(\proppart(\tilde{h}))) \, \text{ and } \, \deg_t(\polypart(h)) \le \deg_t(\polypart(\tilde{h})). $$
\end{corollary}

\begin{proof}
	Since $f=\Delta(\tilde{g})+\tilde{h}$ and $h$ is the remainder of $f$, we have $h-\tilde{h} \in \Delta(F(t))$. Thus
	\begin{equation}\label{Equ:ProperPolyRel}
	\proppart(h)-\proppart(\tilde{h}) \in \Delta(F(t)_{(r)}) \quad \text{and} \quad \polypart(h)-\polypart(\tilde{h}) \in \Delta(F[t])
	\end{equation}
	by Lemma \ref{LM:closure}. The minimality for both components can be shown as follows.
	By Proposition \ref{Prop:decomposeProperpart} we have that $\proppart(\tilde{h})-u \in \Delta(F(t)_{(r)})$ for some $u\in U_S$ with $\deg_t(\den(u)) \le \deg_t(\den(\proppart(\tilde{h})))$. Thus $\proppart(h)-u \in \Delta(F(t)_{(r)})$ by~\eqref{Equ:ProperPolyRel}. Since $\proppart(h),u\in U_S$, $\proppart(h)-u\in U_S$. Hence $\proppart(h)-u \in \Delta(F(t)_{(r)}) \cap U_S$. 
	 So $\proppart(h)-u=0$ by Lemma~\ref{LM:proppart}, which implies that $\deg_t(\den(\proppart(h)))=\deg_t(\den(u))\leq \deg_t(\den(\proppart(\tilde{h})))$.\\	
	Similarly, by Proposition \ref{PROP:complement}, there is $v \in V_{\theta}$ with $\deg_t(v)\le \deg_t(\polypart(\tilde{h}))$ such that $\polypart(\tilde{h})-v \in \Delta(F[t])$. 
	So $\polypart(h)-v \in \Delta(F[t])$ by~\eqref{Equ:ProperPolyRel}. In addition, $\polypart(h) \in V_{\theta}$ implies $\polypart(h)-v \in \Delta(F[t]) \cap V_{\theta}=\{0\}$. Hence, $\polypart(h)-v=0$ which implies $\deg_t(\polypart(h))=\deg_t(v) \le \deg_t(\polypart(\tilde{h})).$
\end{proof}

\section{Complete reductions in towers of $\Sigma^*$ extensions}\label{SECT:towers}

This section divides into three parts. In Section \ref{SUBSECT:recursiveMethod}, we construct complete reductions in towers of $\Sigma^*$-extensions (also called $\Sigma^*$-towers) and elaborate on the algorithmic aspects of the obtained recursive method. In particular, an experimental comparison
between the complete reduction and the built-in algorithm of the package {\tt Sigma}  in Mathematica is given. In Section~\ref{Sec:PT} we illustrate how complete reductions can be utilized to solve the parameterized telescoping problem in such $\Sigma^*$-towers. Finally, we connect to Karr's reduced $\Sigma^*$-extensions in Section~\ref{SUBSECT:wgt} and show how complete reductions can be used to reduce the nesting depth of input sums. 

\subsection{A recursive telescoping algorithm based on complete reductions}\label{SUBSECT:recursiveMethod}

A tower of $\Sigma^*$-extensions can be introduced as follows.

\begin{definition} \label{DEF:tower}
	Let $(F,\sigma)$ be a difference field with constant field $C$ and let $(E,\sigma)$ be a tower of $\Sigma^*$-extensions of $(F,\sigma)$, i.e, 
\begin{equation} \label{EQ:tower}
	\begin{array}{cccccccc}
		F=F_0       & \leq &  F_1      & \leq & \cdots         & \leq     & F_n=E  \\
		&         & \shortparallel &         &            &             & \shortparallel \\
		&         &  F_0(t_1) &         &                   &  & F_{n-1}(t_n)
	\end{array}
\end{equation}
	is a tower of field extensions where for all $1\leq i\leq
	n$ each $F_i=F_{i-1}(t_i)$ is a transcendental field
	extension of $F_{i-1}$, and $\sigma(t_i)=t_i+a_i$ with $a_i\in F_{i-1}$ and $C_{E}=C$.
	In other words, $t_i$ is a $\Sigma^*$-monomial over $F_{i-1}$ for $1\leq i\leq n$. We call $(t_1,\dots,t_n)$ also a {\em tower of $\Sigma^*$-monomials} over $(F,\sigma)$ and  $(E,\sigma)$ a {\em$\Sigma^*$-tower} over $(F,\sigma)$.
\end{definition}

\begin{theorem}\label{TH:CR}	
	Let $(F_n,\sigma)$  be a $\Sigma^*$-tower over $(F_0,\sigma)$ with constant field $C$ as given in Definition~\ref{DEF:tower} and assume that there is a complete reduction $\phi=\phi_0: F_0 \rightarrow F_0$ for $\Delta(F_0)$. Furthermore, let $\Theta_0$ be a $C$-basis of $F_0$ and consider the  $\Theta_0$-canonical basis $\Xi$ given in~\eqref{Equ:FullCThetaBasis}. 
	 Then for each $i$ with $0 \le i \le n$, there is a complete reduction $\phi_i$ for $\Delta(F_i)$ with a second pair $(\theta_i, c_i)\in\Xi\times C$ associated to $F_{i-1}(t_i)$ with $\ind_n(\theta_i)=\ind_{n}(\phi_{i-1}(\Delta(t_{i})))$. 
\end{theorem}

\begin{proof}
	We proceed by induction on $n$. The statement clearly holds for $n=0$. Assume that there is a complete reduction $\phi_{n}$ for $\Delta(F_{n})$ with $n\neq0$ as claimed in the theorem and consider the $\Sigma^*$-monomial $t_{n+1}$ over $(F_n,\sigma)$.
	It follows by Lemma~\ref{Lemma:EffectiveBasisLifting} that there exists a second pair $(\theta_{n+1}, c_{n+1})$ associated to $F_{n+1}$ with $\ind_{n+1}(\phi_{n}(\Delta(t_{n+1})))=\ind_{n+1}(\theta_{n+1})$. Consequently, we obtain a complete reduction $\phi_{n+1}$ for $\Delta(F_{n+1})$ by replacing $F$ with $F_{n}$ and $\phi$ with $\phi_{n}$ in Theorem~\ref{TH:decomp}.
\end{proof}

The algorithmic version can be summarized in the following theorem if one can solve Problem~SE in each of the arising $\Sigma^*$-monomials $t_i$ over $(F_0,\sigma)$. Precisely this can be accomplished if the ground field $(F_0,\sigma)$ is $\sigma^*$-computable. We omit the technical details given in~\cite[Definition~1]{KS:2006} and mention only that such difference fields cover Karr's general $\Pi\Sigma^*$-fields or $\Pi\Sigma^*$-field extensions over the free difference field~\cite{KS:2006}. In this article we will only exploit a very special case of such a $\sigma^*$-computable difference field $(F_0,\sigma)$: $F_0=C$  and the field of constants $C$ is given by a rational function field over an algebraic number field. 

\begin{theorem}\label{Thm:CRAlgorithmDF}
	Let $(F_n,\sigma)$  be a $\Sigma^*$-tower over $(F_0,\sigma)$ with constant field $C$  as given in Definition~\ref{DEF:tower} and assume that $(F_0,\sigma)$ is computable, has an effective $C$-basis, and is $\phi_0$-computable w.r.t.\ a complete reduction $\phi=\phi_0: F_0 \rightarrow F_0$ for $\Delta(F_0)$. Then for all $1\leq i\leq n$, $F_i$ is $\phi_i$-computable with the complete reduction $\phi_i$ for $\Delta(F_i)$ as given in Theorem~\ref{TH:CR} if one of the following holds.
	\begin{enumerate}
	 \item One can solve Problem~SE in $F_{i-1}[t_i]$ for each $1\leq i\leq n$. 
	  \item $(F,\sigma)$ is $\sigma^*$-computable. 
	  \item $F=C$ is a rational function field over an algebraic number field.
	\end{enumerate}	
\end{theorem}
\begin{proof}
If $n=0$ the statement clearly holds. Otherwise, we observe that by Lemma~\ref{Lemma:EffectiveBasisLifting} a second pair in the proof of Theorem~\ref{TH:CR} can be computed explicitly.
Moreover, Property~3 implies Property~2 by~\cite[Thm.~3.5]{Schneider:2005} and Property~2 implies Property~1 by~\cite[Cor.~1]{KS:2006}. Finally, if Property~1 holds, the construction of Theorem~\ref{TH:decomp} is algorithmic (see Algorithm \texttt{CompleteReduction}). Thus also the induction step in the proof of Theorem~\ref{TH:CR} can be carried out explicitly which proves the statement.
\end{proof}

When one applies Algorithm \texttt{CompleteReduction} to $f$ in $(F_n,\sigma)$ one enters Algorithm \texttt{\tt ReductionForPolynomials}($F_n,f_2$) in Step~3. 
There one executes in Step~1 the command \texttt{AuxiliaryReduction}($p$) and restarts in the while loop the Algorithm \texttt{CompleteReduction} in the field $F_{n-1}$ below. Similarly, one enters also the sub-algorithm \texttt{EchelonBasis} which inside executes again \texttt{CompeteReduction} in the field $F_{n-1}$ below. Summarizing, the machinery is highly recursive and a tree of reductions is generated that call several instances of \texttt{CompeteReduction} in each extension level $F_i$ for $1\leq i\leq n$.

The following technical aspects are in place.
\begin{itemize}
\item When executing {\tt ReductionForProperRationalFunctions} in $(F_i,\sigma)$ within the recursive calls, one has to take care that the set $S_i$ of representants of shift equivalent factors used for the definition of~\eqref{EQ:ProperCompl} (with $S=S_i$) is fixed during all recursive calls. As explained in Remark~\ref{Remark:ExpandSetS} we initialize these sets $S_i$ for $1\leq i\leq n$ to the empty set and append new representant to $S_i$ whenever a new element $p$ in Algorithm \texttt{ReductionForProperRationalFunctions} arises.

\item  The determination of a second pair is essential for the construction of the complementary space given in Definition~\ref{Def:ComplementSpaceForPoly}. Since the second component of the first pair is needed to compute the second pair, it is natural to fix the first pair also in advance. Thus both pairs are fixed in a preprocessing step for each arising $\Sigma^*$-monomial $t_i$. This leads to the following extra improvement.

\item Using always the same first pair of a $\Sigma^*$-monomial $t_i$, also the echelon basis induced by this first pair is fixed. In particular, one can store the basis elements for each $t_i$ in a separate list $L_i$ and can reuse the basis elements or can enlarge the list with extra basis elements with higher degrees whenever this is necessary within the recursion. 
\end{itemize}


Subsequently we make the following notational convention for the rest of this section.

\begin{convention}\label{CON:cr}
	Let $(F_n,\sigma)$  be a $\Sigma^*$-tower over $(F_0,\sigma)$ as given in \eqref{EQ:tower}, and $\phi_0$ be a complete reduction on $F_0$ for $\Delta(F_0)$.  
	Let $\Xi$ be the effective $C$-basis of $F_n$ given in~\eqref{Equ:FullCThetaBasis}. 
	For all $i$ with  $1 \le i \le n$,
	$\phi_i: F_i \rightarrow F_i$ stands for the complete reduction
	for $\Delta(F_i)$, and $\left(g_{t_i}, \phi(\Delta(t_i)) \right)$ and $\left(\theta_i, c_i\right)$ for fixed first and second pairs associated to $F_i$, respectively.
\end{convention}

\bigskip

\begin{corollary}\label{COR:ind}
	Let $f \in F_n$. 
	Then $\ind_n(\phi_n(f)) \le \ind_n(f)$. 
\end{corollary}
\begin{proof}
	Let $m=\ind_n(f)$. Then $f\in F_m$ but $f \notin F_{m-1}$.  It follows that $\phi_m(f) \in F_m$. So $\ind_n(\phi_m(f)) \le \ind_n(f)$. By Corollary \ref{COR:remainder}.1, we have 
	\begin{equation}\label{Equ:phim+1Id}
	\phi_{m+1}(f)=\phi_m(f)+c_{m}\phi_{m}(\Delta(t_{m+1})),
	\end{equation}
	 where 
	$c_m=-\theta_{m+1}^*(\phi_m(f))/\theta_{m+1}^*(\phi_m(\Delta(t_{m+1})))$. Write $\phi_m(f)=\sum_{w_i \in \Xi} c_iw_i$ in the effective $C$-basis $\Xi$ given in Convention \ref{CON:cr}.
	If $\phi_m(f)$ is free of $\theta_{m+1}$, then $c_m=0$. So $\phi_{m+1}(f)=\phi_m(f)$, in particular $\ind_n(\phi_{m+1}(f))=\ind_n(\phi_m(f))$. Otherwise, we conclude with $\ind_n(\phi_m(f))=\max\{\ind_n(w_i)\mid c_i \neq 0 \}$ by~\eqref{Equ:IndToBasisEl} and the property that the coefficient of $\theta_{m+1}$ in $\phi_m(f)$ is nonzero that $\ind_n(\theta_{m+1})\leq \ind_n(\phi_m(f))$. With
	Theorem~\ref{TH:CR} we have $\ind_n(\phi_{m}(\Delta(t_{m+1})))=\ind_n(\theta_{m+1})$ and thus 
	$\ind_n(\phi_{m}(\Delta(t_{m+1})))\leq \ind_n(\phi_m(f))$. Finally, with~\eqref{Equ:phim+1Id} it follows that
	$$\ind_n(\phi_{m+1}(f))\le \max \{\ind_n(\phi_m(f)), \ind_n(\phi_{m}(\Delta(t_{m+1}))) \}=\ind_n(\phi_m(f)).$$ 
	So in any case we have that $\ind_n(\phi_{m+1}(f))\leq \ind_n(\phi_m(f))$. 
	Repeating the above analysis in a finite number of steps, we conclude that 
	$$\ind_n(\phi_{n}(f)) \le \ind_n(\phi_{n-1}(f)) \le \cdots \le \ind_n(\phi_{m+1}(f)) \le\ind_n(\phi_m(f)) \le \ind_n(f).$$
Note that $f\in F_m \setminus F_{m-1}$, i.e., $f$ depends on $t_m$. By construction, $\phi_m(f)\in F_m$ and thus the last inequality follows.
\end{proof}

\begin{example}\label{EX:HN}
	We try to simplify the sum 
	$$\sum_{k=1}^{n}\frac{1}{k}\sum_{j=1}^{k}\frac{H_j}{j}.$$
	Let $F_0=\bQ(x)$ with $\sigma(x)=x+1$ and take the $\Sigma^*$-tower $F_2=F_0(t_1,t_2)$ over $F_0$, where 
	$$\sigma(t_1)=t_1+\frac{1}{x+1} \quad \text{and} \quad \sigma(t_2)=t_2+\sigma\left(\frac{t_1}{x}\right)=t_2+\frac{(x+1)t_1+1}{(1+x)^2}.$$
	To $F_1$, we associate 
	$$(g_{t_1},\phi_{0}(\Delta(t_1)))=\left(\frac{1}{x}, \frac{1}{x}\right), \quad (\theta_1, c_1)=\left(\frac{1}{x},1 \right)$$
	and to $F_2$, we associate
	$$(g_{t_2},\phi_{1}(\Delta(t_2)))= \left(\frac{1+x^2t_1^2}{2x^2},\frac{1}{2x^2}\right), \quad (\theta_2,c_2)=\left( \frac{1}{x^2},\frac{1}{2}\right).$$
	
	Here we choose $f:=t_2/x$ as the element in $F_2$ that represents $\frac{1}{k}\sum_{j=1}^{k}H_j/j$. Applying Algorithm {\tt CompleteReduction} to $f$ yields the $\Sigma$-pair
	$$(g,\,r)=\left( \frac{3x^3t_1t_2-x^3t_1^3-3x^2t_2+1}{3x^3},\frac{1}{3x^3}\right)$$
	of $f$. Since $r$  is nonzero, $f$ is not summable in $F_2$. In particular, using Corollary~\ref{Cor:SigmaCheck} it follows that we can adjoin the $\Sigma^*$-monomial $t_3$ to $F_2$ with $\Delta(t_3)=r$ which gives $f=\Delta(g+t_3)$. Note that this relation is reflected by the identity
	\begin{equation}\label{EQ:identity1}
		\sum_{k=1}^{n}\frac{1}{k}\sum_{j=1}^{k}\frac{H_j}{j}=H_n \sum_{k=1}^{n} \frac{H_k}{k}-\frac{H_n^3}{3}+\frac
		{H_n^{(3)}}{3}.
	\end{equation}

For concrete problem solving (see, e.g.~\cite{Schneider:2021,PS2003,BMSS:22b}) one can often restrict to the following special case.
Let $(F_n,\sigma)$ be as above a $\Sigma^*$-tower over $(F_0,\sigma)$ with the restriction that $\Delta(t_i)\in F_0[t_1,\dots,t_{i-1}]$ for $1\leq i\leq n$, i.e., inside of sums only polynomial expressions of sums arise. Furthermore define the subring of polynomials $E_i=F_0[t_1,\dots,t_i]$. Then it follows straightforwardly that $\sigma$ restricted to $E_i$ forms a ring automorphism. In short, $(E_n,\sigma)$ is a difference ring (which in~\cite{Schn2016a,Schn2017} are also called $\Sigma^*$-ring extensions). 
Observing that $\Delta(E_i)$ is a $C$-subspace of $F_i$ we get the following connection.

\begin{corollary}
$\phi_i|_{E_i}$ is a complete reduction of $\Delta(E_i)$ on $E_i$ for $1\leq i\leq n$. 
\end{corollary}
\begin{proof}
Suppose that there is an $i$ with $1\leq i\leq n$ such that $\phi'_i:=\phi_i|_{E_i}$ is not a complete reduction of $\Delta(E_i)$ on $E_i$. Let $i\geq1$ be minimal with this property. Then there is an $f\in E_i$ such that $\phi'_i(f)\notin E_i$ or we have $\phi'_i(f)\in E_i$ but there is no $g\in E_i$ with $f=\Delta(g)+\phi'_i(f)$. 
The minimality of $i$ implies $\im(\phi'_{i-1})=\im(\phi_{i-1}|_{E_{i-1}})\subseteq E_{i-1}$. For each $p\in E_i$, there is $q \in E_i$ such that $p-\Delta(q)\in \bigoplus_{i\in \bN} \im(\phi_{i-1}')t^i$ by replacing $F$ with $E_{i-1}$, $t$ with $t_i$ and $\phi$ with $\phi_{i-1}$ in Proposition \ref{PROP:auxiliarySpace}. So
 we conclude that $\phi'_i(f)\in \bigoplus_{i\in \bN} \im(\phi_{i-1}')t^i$ by Corollary~\ref{COR:remainder}.2. Thus $\phi'_i(f)\in E_i$ and it follows that $g\in F_i\setminus E_i$ such that $f=\Delta(g)+\phi'_i(f)$. However, $\Delta(g)=f-\phi'_i(f)\in E_i$. With~\cite[Thm.~2.7]{Schneider:2010} we conclude that $g\in E_i$, a contradiction. 
\end{proof}


Even more, one can streamline Algorithm~\texttt{CompleteReduction} for such a difference ring $(E_n,\sigma)$ by observing that execution {\tt ReductionForProperRationalFunctions} is obsolete and one simply sets $(g,h)=(0,0)$. This follows by induction on $n$. Namely, suppose that the modified algorithm is correct with less than $n$ $\Sigma^*$-monomials and executes the modified algorithm \texttt{CompleteReduction} in $E_n$. Then by Corollary~\ref{COR:remainder}.2 we may set $(g,h)=(0,0)$. Further we observe that when entering the steps for \texttt{ReductionForPolynomials} one only carries out operations in the difference ring $(E_{n},\sigma)$ leading only to expressions in $E_n$. Thus whenever one calls \texttt{CompeteReduction} in $E_{n-1}$ (in \text{AuxilaryReduction} and \text{EchelonBasis}) the argument is again a polynomial in $E_{n-1}$. Thus by induction one never deals with proper rational function contributions in $(F_{i-1}(t_i))_{(r)}$ and the claim is certified. As a consequence, solving Problem~SE is obsolete and Theorem~\ref{TH:decomp} simplifies as follows.

\begin{corollary}\label{Cor:SringExt}
Let $(E_n,\sigma)$  be a tower of $\Sigma^*$-ring extensions over $(F_0,\sigma)$ with constant field $C$  as given above and assume that $(F_0,\sigma)$ is computable, has an effective $C$-basis, and is $\phi_0$-computable w.r.t.\ a complete reduction $\phi_0$ for $\Delta(F_0)$ on $F_0$. Then for all $1\leq i\leq n$, $E_i$ is $\phi_i$-computable with the complete reduction $\phi_i$ for $\Delta(E_i)$ as given in Theorem~\ref{TH:CR}.
\end{corollary}

\subsubsection*{Implementation aspects}

We have implemented  Algorithm {\tt CompleteReduction}~({\tt CR}) in the computer algebra system Mathematica using subroutines of {\tt Sigma} in the setting of $\Pi\Sigma^*$-fields to compute, e.g., the $\sigma$-factorization needed in \texttt{ReductionForProperRationalFunctions}. In the following we will compare it with the function {\tt RefinedTelescoping}~({\tt RT}) in the summation package {\tt Sigma} that implements the algorithms given in~\cite{Schn2008,Schn2007,Schneider:2015,Schn2016a}.
In our case study we take the $\Sigma^*$-tower $(F_2,\sigma)$ over $(\bQ,\id)$ with $F_2=\bQ(x)(t_1)(t_2)$, where
$$\sigma(x)=x+1, \quad \sigma(t_1)=t_1+\frac{1}{x+1},\quad \text{and} \quad \sigma(t_2)=t_2+\frac{1}{(x+1)^2}.$$
Note that $t_1$ and $t_2$ represent $H_n$ and $H_n^{(2)}$, respectively.

In each test suite we chose for $i\in\bN$ three random polynomials $p_i\in\bQ(x)[t_1,t_2]$ where the selected generators have total degree $i$ (details are given below) and measured the average time to solve the telescoping equation $\Delta(g)=f_i$ for the given summand $f_i=\Delta(p_i)\in \bQ(x)[t_1,t_2]$, i.e., the time needed to reconstruct $p_i$ from $f_i$.
More precisely, we applied for the generators $t_2$ and $t_1$ to the simplified version given in Corollary~\ref{Cor:SringExt} (i.e., {\tt ReductionForProperRationalFunctions} is not used), but used the full machinery of {\tt CR} in the difference field $(\bQ(x),\sigma)$ which reduces to Abramov's algorithm~\cite{Abra1975}; compare Corollary~\ref{COR:ADRational}. Similarly, we exploit other refined algorithms in \texttt{Sigma} for this special structure when applying {\tt RT}.

All timings are measured in seconds on a computer with Linux, CPU 3.00 GHZ, Intel Core i7-9700, 32G memory. The average timings are summarized as follows.\\
$\bullet$ In the first suite, we generated randomly $p_i \in \bQ[x,t_1,t_2]$ where the total degree of all three variables is $i$.
\begin{center}
	\begin{tabular}{|c|c|c|c|c|c|c|c|c|c|} \hline
		$i$ & $10$ & $15$ &  $20$ & $25$ & $30$ & $35$ & $40$ & $45$ & $50$   \\ \hline
		{\tt RT}  & ~$1.74$~& ~$4.55$~ &  ~$11.82$~  & ~$23.73$~  &  ~$66.96$~  &  ~$143.87$~ &  ~$207.23$~&~$368.09$~&~$547.52$~ \\ \hline
		{\tt CR}  &  1.19 &  1.69 &   3.68 &   3.04 &   6.19  &   13.40  &  18.03 & 32.04 & 49.57  \\ \hline
	\end{tabular}
\end{center}
$\bullet$ In the second suite, we picked polynomials $p_i \in \bQ(x)[t_1,t_2]$ where the total degree of the polynomials in $t_1,t_2$ is $i$ and the coefficients are quotients of random polynomials in $\bQ[x]$ with degree $5$.
\begin{center}
	\begin{tabular}{|c|c|c|c|c|c|c|c|c|} \hline
		$i$  & $11$ & $12$ & $13$ &  $14$ & $15$ & $16$ & $17$ & $18$   \\ \hline
		{\tt RT} & ~428.15~ & ~650.26~ & ~1428.50~  & ~1928.40~  & ~4191.05~ & ~6650.85~ &  ~$>36000$~ &  ~$>36000$~\\ \hline
		{\tt CR}  & 59.19 & 81.68   &  127.28 &   163.28 &   303.58 &   381.82  &   431.12  &   681.55  \\ \hline
	\end{tabular}
\end{center}

The timings show that {\tt CR} outperforms {\tt RT} in a small $\Sigma^*$-tower where the input polynomials have high degrees. At this point one may remark that \texttt{Sigma}'s algorithms are tailored for big $\Pi\Sigma^*$-towers with up to 100 generators (where the degrees of the $\Sigma^*$-monomials in the summand are not large) within the setting of depth-optimal $\Pi\Sigma^*$-extensions~\cite{Schn2008}. Further case studies will be necessary (with further extensions of our new approach) to be able to compare also such scenarios fairly.


\end{example}

\subsection{Parameterized telescoping}\label{Sec:PT}

	A key problem of symbolic summation is the task to solve parameterized telescoping equations. In the difference setting this can be formalized as follows. Given $f_1,\dots,f_m \in F_n$, find $c_1,\dots,c_m\in C$, not all zero, and $g\in F_n$ such that
	\begin{equation} \label{EQ:PT}
		c_1f_1+\cdots+c_nf_m =\Delta(g),
	\end{equation}
	holds; it reduces to the telescoping problem for the special case $m=1$, but also covers Zeilberger's creative telescoping paradigm~\cite{Zeil1990b,Zeil1991,PWZ1996} in the setting of difference fields~\cite{Schn2001,PS2003}. More generally,	
	define for $\vf=(f_1,\cdots,f_m) \in F_n^{m}$ the solution set~\cite{Karr1981}	
	$$V(\vf,F_n)=\{(c_1,\cdots,c_m,g)\in C^m \times F_n\mid c_1f_1+\cdots+c_nf_m =\Delta(g)\}.$$
	Since this set forms a $C$-linear space with dimension no more than $m+1$ (see~\cite[Cor.~3.1.1]{Schn2001}), the task to solve the parameterized telescoping problem is covered by computing a basis of $V(\vf,F_n)$. In~\cite{PS2003,Schn2005,Schn2008,Schneider:2015,Schn2017} various improvements have been elaborated how the parameterized telescoping can be tackled in the difference ring and field approach.
	
	As illustrated in~\cite{BCCLX2013,CKK2016,CDK2021,BCPS2018,CDK2023}
in symbolic integration, reduction based methods can be used to significantly speed up the task to solve the parameterized telescoping problem. Using our reduction based approach for $\Sigma^*$-extensions these ideas can be exploited as follows.

	For $1\leq i\leq m$, we compute the $\Sigma$-pairs $(g_i,\phi_n(f_i))$ of $f$ w.r.t.\ the complete reduction $\phi_n$, i.e., we get
	$$f_i=\Delta(g_i)+\phi_n(f_i).$$	
	Now the crucial observation is that applying $\phi_n$ to~\eqref{EQ:PT} yields 
	\begin{equation}\label{Equ:CConstraintForPT}
	c_1 \phi_n(f_1)+\cdots+c_m\phi_n(f_m)=0
	\end{equation} 
	for the given $\phi_n(f_i)$. This leads to a linear algebraic system $L$ in $c_1, \ldots, c_m$ and the task to determine the $(c_1,\dots,c_m)\in C^m$ reduces to the problem to solve the system $L$ by linear algebra methods. Suppose that we get the $C$-basis  
	$B:=\{(c_{i,1},\cdots,c_{i,m})\mid 1\le i \le j\}$ of $L$ for some $1\leq j\le m$. Then this leads immediately to the desired $C$-basis 
	$$\{(0,\dots,0,1)\}\cup\{(c_{i,1},\cdots,c_{i,m},w_i)\mid 1\le i \le j\}$$
	of $V(\vf,F_n)$ with
	$$w_i=\sum_{\ell=1}^{m}c_{i,\ell}g_{\ell}\in F_n.$$ 
	We remark that for the standard approach in difference rings and fields~\cite{PS2003,Schn2005,Schn2008,Schneider:2015,Schn2017} 
	it is often a challenge to compute the expressions $w_i\in F_n$; they usually blow up if one brings them over a common denominator. Using complete reductions as a preprocessing step to compute the $\Sigma$-pairs  $(g_i,\phi_n(f_i))$ and to solve afterward the underlying linear system~\eqref{Equ:CConstraintForPT} to determine the $c_i$, makes the reduction based methods so interesting and will be explored further in future investigations.
 	
	\begin{example}\label{EX:PT}
		For the $\Sigma^*$-tower $F_2$ over $F_0$ given in Example \ref{EX:HN}, let $\vf=(f_1,f_2,f_3) \in F_2^{3}$, where
		$$ f_1=\frac{1+t_1-t_2-xt_2}{(1+t_1)(1+x)}, \quad f_2=\frac{xt_1+t_1-x}{(xt_1+t_1+x)t_1} \quad \text{and} \quad f_3=\frac{3t_2}{1+t_1}.$$
		With the Algorithm {\tt CompleteReduction} we compute $\Sigma$-pairs 
		$$\left( t_1,\, -\frac{t_2}{t_1+1}\right), \quad \left( \frac{x}{t_1},\, 0 \right), \quad \text{and} \quad \left( 0,\, \frac{3t_2}{t_1+1}\right)$$
		of $f_1$, $f_2$ and $f_3$ w.r.t.\ $\phi_2$, respectively.  Thus there exist $c_1,c_2,c_3 \in C$ such that 
		$$c_1f_1+c_2f_2+c_3f_3=\Delta(g)$$ 
		for some $g\in F_2$ if and only if 
		$$c_1\frac{-t_2}{(1+t_1)}+c_2\frac{3t_2}{(1+t_1)}=0.$$ 
		From the linear system we compute the $C$-basis $\{(3,0,1),\,(0,1,0)\}$. Therefore, $\left\{\left(0,0,0, 1\right), \left(3,0,1, 3t_1\right),\,\left(0,1,0,x/t_1\right) \right\}$ is a $C$-basis of $V(\vf,F_2)$.
	\end{example}
	
\begin{example}	
Given $f(n,k)=\frac{H_k}{n-k+1}$,
we are looking for $c_1,c_2,c_3\in\bQ(n)$ and $g(n,k)$ such that the creative telescoping equation
\begin{equation}\label{Equ:CreaSol}
c_1\,f(n,k)+c_2\,f(n+1,k)+c_3\,f(n+2,k)=g(n,k+1)-g(n,k)
\end{equation}
holds for all $n\geq0$ and $0\leq k\leq n$; the left-hand side is also called a telescoper and $g(n,k)$ is called the certificate.
	For this task, we take the rational function field $C=\bQ(n)$ and construct the $\Sigma^*$-tower $(F_2,\sigma)$ over $(C,\id)$ with $F_2=C(x)(t_1)$ where $\sigma(x)=x+1$ and $\sigma(t_1)=t_1+1/(x+1)$. In $F_2$ we can represent $f(n,k)$ by
	$f_1=t_1/(n-x+1)$, $f(n+1,k)$ by $f_2=t_1/(n-x+2)$ and $f(n+2,k)$ by $f_3=t_1/(n-x+3)$.
	Let $\phi$ be the complete reduction for $\Delta(F_2)$ on $F_2$. Since $f_1\in C(x)[t_1]$ whose the coefficient w.r.t. $t_1$ is $\sigma$-simple and free of $1/x$, we have $\phi(f_1)=f_1$ and we get the $\Sigma$-pair $(g_1,\phi(f_1))=(0,f_1)$. Applying Algorithm {\tt CompleteReduction} to $f_2$ and $f_3$ produces the $\Sigma$-pairs 
	\begin{align*}
		(g_2,\phi(f_2)) &=\left( -\tfrac{xt_1}{(n+2)(n-x+2)},\,\tfrac{t_1}{n-x+1}+\tfrac{1}{(n+2) (n-x+1)}\right),\\
		(g_3, \phi(f_3))&=\left(\tfrac{t_1}{-n+x-3}+\tfrac{n t_1+3 t_1+1}{(n+3) (-n+x-2)}+\tfrac{(2 n+5) t_1}{(n+2) (n+3)},\tfrac{t_1}{n-x+1}+\tfrac{2 n+5}{(n+2) (n+3) (n-x+1)}\right).
	\end{align*}
	By linear system solving we find $c_1=-n-2,\,c_2=2n+5,\, c_3=-n-3$ such that 
	$$c_1\phi(f_1)+c_2\phi(f_2)+c_3\phi(f_3)=0.$$
	Taking $g=c_1g_1+c_2g_2+c_3g_3=\frac{n t_1+2 t_1-1}{x-n-2}-\frac{n t_1+3 t_1}{x-n-3}$ solves~\eqref{EQ:PT} with $m=3$. Reinterpreting the $g$ as $g(n,k)$ in terms of our given summation objects provides a solution of~\eqref{Equ:CreaSol} and summing this equation over $k$ from $0$ to $n$ and taking care of compensating terms we obtain the linear recurrence	
	$$(-n-2)S(n)+(2n+5)S(n+1)+(-n-3)S(n+2)=-\tfrac{2}{n+2}$$
	for the sum $S(n)=\sum_{k=1}^nf(n,k)$.
\end{example}

\subsection{Well generated and reduced towers of $\Sigma^*$-extensions}\label{SUBSECT:wgt}
Following~\cite{Karr1981} a tower of $\Sigma^*$-extensions $(F_n,\sigma)$ over $(F_0,\sigma)$ with~\eqref{EQ:tower} is said to be {\em reduced over $F_0$}, if for each $\Sigma^*$-monomial $t_i$  with $1 \le i \le n$ the following property holds: if there exist $g_i\in F_{i-1}$ and $r_i\in F_0$ such that  
\begin{equation}\label{Equ:RefinedTeleReduced}
\Delta(t_i)=\Delta(g_i)+r_i,  
\end{equation}
then $\Delta(t_i)\in F_0$. 
Reduced $\Sigma^*$-towers can be useful to find closed forms of given sums by using the following result;
see~\cite{Karr1981} (compare also \cite[Theorem 4.2.1]{Schn2001} and \cite[Theorem 9]{Schn2010}). 

\begin{theorem}\label{Thm:KarrFT}
	(Karr's fundamental theorem)~Let $(F_n,\sigma)$ be a reduced $\Sigma^*$-tower over $(F_0,\sigma)$. Assume that $f\in F_0$ and $g\in F_n$ such that $f=\Delta(g)$. Then
	$$g=\sum_{i\in S}c_it_i+w$$
	where $w\in F_0$, $c_i\in C$ and 
	$$S=\{ i\mid \Delta(t_i) \in F_0\}.$$
\end{theorem}

Given a $\Sigma^*$-tower, there are algorithms available in~\cite{Schn2004,Schneider:2015}, to decide algorithmically if for given $\Delta(t_i)\in F_n$ there exist $g_i\in F_n$ and $r_i\in F_0$ such that~\eqref{Equ:RefinedTeleReduced} holds. In particular, this problem has been solved in~\cite[Thm.~6.1]{Schn2007} using partially reduction-based ideas.
Then given any of these algorithms and using \cite[Algorithm 1]{Schn2010}, one can compute a reduced $\Sigma^*$-tower which is isomorphic to the input tower. These algorithms rely on recursive reductions where one has to solve parameterized telescoping problems in the fields below. 
With our reduction based methods presented above we can solve this problem efficiently without solving explicitly any linear difference equation in the subfields below. 

In the following we make this result more precise.
First we refine $\Sigma^*$-towers using the construction of complete reductions and show afterwards that they are reduced.

\begin{definition}
	Let $(F_n,\sigma)$ with $F_n:=F_0(t_1,\cdots,t_n)$ be a $\Sigma^*$-tower over $(F_0,\sigma)$. If $\Delta(t_i)\in\im(\phi_{i-1})$ for each $i$ with $1\le i\le n$, then  $F_n$ is said to be {\em well generated}.
\end{definition}

\begin{lemma}
	Any well generated $\Sigma^*$-tower $(F_n,\sigma)$ over $(F_0,\sigma)$ is reduced over $F_0$.
\end{lemma}
\begin{proof} 
Suppose that $F_n$ is not reduced. Then we can take an $i$ with $1 \le i \le n$  and $\Delta(t_i) \notin F_0$ such that there are $g_i \in F_{i-1}$ and $r_i\in F_0$ with~\eqref{Equ:RefinedTeleReduced}. Since $\Delta(t_i)\in\im(\phi_{i-1})$ and $\phi_{i-1}$ is idempotent, it follows that $\phi_{i-1}(\Delta(t_i))=\Delta(t_i)$. Hence applying $\phi_{i-1}$ to both sides of~\eqref{Equ:RefinedTeleReduced}, it follows that $\Delta(t_i)=\phi_{i-1}(\Delta(g_i))+\phi_{i-1}(r_i)=\phi_{i-1}(r_i)$. Moreover, $\ind_{i-1}(\phi_{i-1}(r_i)) \le \ind_{i-1}(r_i)$ by Corollary \ref{COR:ind}. So $\ind_{i-1}(\Delta(t_i)) \le \ind_{i-1}(r_i)$, which implies that $\Delta(t_i)\in F_0$, a contradiction. 
\end{proof}

For difference fields $(F,\,\sigma)$ and $(K,\,\tilde{\sigma})$ a map $\tau:F \rightarrow K$ is called a {\em difference field isomorphism}  if $\tau$ is a field isomorphism with $\tilde{\sigma}(\tau(a))=\tau(\sigma(a))$ for all $a\in F$. In this case we say that $(F,\,\sigma)$ and $(K,\,\tilde{\sigma})$ are {\em isomorphic}.

Next, we show how a $\Sigma^*$-tower can be transformed (via a difference ring isomorphism) to a well generated version; compare~\cite[Algorithm~1]{Schn2010}.

\begin{proposition}\label{PROP:diffiso}
	Let $(F_n,\sigma)$ be a $\Sigma^*$-tower over $(F_0,\sigma)$. Then there exists a well generated $\Sigma^*$-tower $(K_n,\sigma)$ over $(F_0,\sigma)$ with $K_n=F_0(u_1,\cdots,u_n)$ 
	and a difference field isomorphism $\tau_n$ from $F_n$ onto $K_n$ with $\tau_n |_{F_0}=\id$.
\end{proposition}

\begin{proof}
	We proceed by induction on $n$. For the base case $n=0$ nothing has to be shown.
	Assume that $n>1$ and that the conclusion holds for $n-1$. So there is a difference field isomorphism $\tau_{n-1}$ from $(F_{n-1},\sigma)$ onto  a well generated $\Sigma^*$-tower $(K_{n-1},\tilde{\sigma})$ with $\tau_{n-1} |_{F_0}=\id$. Then $(K_{n-1}(\tilde{u}_n),\tilde{\sigma})$ with $\tilde{\sigma}(u_n)-u_n =\tau_{n-1}(\Delta(t_n))$ is a $\Sigma^*$-extension because $t_n$ is a $\Sigma^*$-monomial over $F_{n-1}$ and $\tau_{n-1}$ is an isomorphism. 
	Moreover,  It follows from \cite[Proposition~18]{Schn2008} that $\tau_{n-1}$ can be extended to a difference field isomorphism from $F_{n-1}(t_n)$ to $K_{n-1}(\tilde{u}_n)$ by sending $t_n$ to $\tilde{u}_n$.
	On the other hand, there is a complete reduction $\tilde{\phi}_{n-1}$ on $K_{n-1}$ for $\Delta(K_{n-1})$ by Theorem \ref{TH:CR}.  Let $(g_n,r_n) $ be a $\Sigma$- pair of $\Delta(\tilde{u}_n)$ with $r_n=\tilde{\phi}_{n-1}(\Delta(\tilde{u}_n))$. It follows that $K_{n-1}(u_n)$  with $\Delta(u_n)=r_n$ is a $\Sigma^*$-extension by Corollary~\ref{Cor:SigmaCheck}. Moreover, there is a difference isomorphism 
\[
\begin{array}{cccc}
	\mu: & K_{n-1}(\tilde{u}_n) & \longrightarrow &  K_{n-1}(u_n)\\
	& \tilde{u} _n   &  \mapsto &  u_n+g_n.      
\end{array}
\]
with $\mu|_{K_{n-1}}=\id$ by \cite[Lemma~21]{Schn2010}. Thus $\mu$ is a difference field isomorphism and therefore $\tau_{n}=\mu \circ \tau_{n-1}$ establishes a difference isomorphism between $F_n$ and $K_n$.
\end{proof}

The following example illustrates that well-generated $\Sigma^*$-towers may help us to reduce the depth of a given input sum.
\begin{example}
Let $F_2:=F_0(t_1,t_2)$ and $f$ be given in Example \ref{EX:HN}. We transform $F_2$ to a well generated one.
We start with the ground field $F_0$ and take the isomorphism $\tau:F_0 \rightarrow F_0$ with $\tau(a)=a$ for all $a\in F_0$. By Example \ref{EX:HN}, we find a $\Sigma$-pair 
$(t_1-1/x,\,1/x)$ of $\Delta(t_1)$. It follows that $F_0(u_1)$ with $\Delta(u_1)=1/x$ is a $\Sigma^*$-extension by Corollary~\ref{Cor:SigmaCheck}. Furthermore, we extend the isomorphism $\tau$ to $\tau:F_0(t_1) \rightarrow F_0(u_1)$ by sending $t_1$ to $u_1+1/x$. Finally, applying Algorithm {\tt CompleteReduction} to $\tau\left( \sigma(t_1/x)\right)$, we obtain the $\Sigma$-pair 
$$\left(\frac{u_1^2}{2}+\frac{u_1}{x}+\frac{1}{x^2},\, \frac{1}{2x^2}\right),$$
which implies that $F_0(u_1,u_2)$ with $\Delta(u_2)=1/2x^2$ is a $\Sigma^*$-tower over $F_0$. In addition, the isomorphism can be further extended to $\tau:F_0(t_1,t_2) \rightarrow F_0(u_1,u_2)$ with $\tau(t_2)=u_2+u_1^2/2+u_1/x+1/x^2$. \\
Consequently, $K_2:=F_0(u_1,u_2)$ is a well generated $\Sigma^*$-tower over $F_0$ and is isomorphic to $F_2$. 
Next, we simplify $\tau(f)$ in $K_2$. The result is
$$\tau(f)=\Delta\left(\frac{u_1^3}{6}+u_1u_2\right)+\frac{1}{3x^3},$$
which produces
$$\sum_{k=1}^{n}\frac1k\sum_{j=1}^{k}\frac{H_j}{j}=\frac{H_n^3}{6}+\frac{H_nH_n^{(2)}}{2}+\frac{H_n^{(3)}}{3}.$$
Compared with the right side of \eqref{EQ:identity1}, we managed to reduce the nesting depth (the number of recursive summation quantifiers) from $2$ to $1$.
\end{example}
 
\section{Conclusion}\label{Sec:Conclusion}

In this article we combined ideas from the differential field~\cite{DGLL2025} and difference field setting~\cite{Schn2007} to derive a new algorithmic approach to solve the refined telescoping equation~\eqref{Equ:RefinedTele} in a tower of $\Sigma^*$-extensions that relies only on complete reductions. In particular, we demonstrated that the complete reduction approach applied to summands in small $\Sigma$-towers (with $\sim3$ generators) but involving big polynomials with large total degree outperforms the telescoping algorithms of {\tt Sigma} that rely on solving telescoping equations in subfields below (by using denominator and degree bounds and computing the solution by recursive system solving). Further experiments (and enhanced implementations) will be necessary to compare these two different approaches for more complicated towers of $\Sigma^*$-extensions.\\
In addition, we illustrated how these reduction based methods can be used to solve parameterized telescoping (including creative telescoping as a special case) problems and to reduce the nesting depth of the input sums. In particular, we connected our constructions  to Karr's reduced $\Sigma^*$-extensions in the context of his Fundamental Theorem (see Theorem~\ref{Thm:KarrFT}). Further developments will be necessary to link complete reductions to more refined structural theorems derived in~\cite{Schn2010}. In particular, it will be interesting to see how depth-optimal $\Sigma^*$-extensions~\cite{Schn2008} can be connected to complete reductions in order to compute sum representations with minimal nesting depth more efficiently.

In general, we developed a general framework that solves the problem to compute complete reductions in a tower of $\Sigma^*$-extensions if the problem can be solved in the ground difference field (see Theorem~\ref{Thm:CRAlgorithmDF} and Corollary~\ref{Cor:SringExt}). This will be the basis for further investigations to extend the existing machinery.
In particular, using results from~\cite{CHKL2015,DHL2018,CDGHL2025} it will be important to develop complete reductions for $\Pi$-extensions and $R\Pi$-ring extensions~\cite{Schn2016a,Schn2017} in which one can model \hbox{($q$--)}hypergeometric products and their mixed versions; for further details we refer to~\cite{Schn2021} and literature therein.

Finally, we remark that the proposed construction of complete reductions are inspired by the differential case and thus these symbolic integration and summation approaches can be compared straightforwardly. In this regard, a natural question is to which extend our results (like reducing the nesting depth of sums in complete reductions) can be taken up to the differential field setting.   
	
\subsection*{Acknowledgment} This research was funded in whole or in part by the Austrian Science Fund (FWF) [Grant-DOI 10.55776/PAT1332123], [Grant-DOI
10.55776/P20347] and [Grant-DOI 10.55776/I6130]. 

S.\ Chen was partially 
supported by the National Key R\&D Programs of
China (No.\ 2020YFA0712300 and No.\ 2023YFA1009401), the NSFC grants (No.\ 12271511 and No.\ 11688101), 
the CAS Funds of the Youth Innovation Promotion Association (No.\ Y2022001), and 
the Strategic Priority Research Program of the Chinese Academy of Sciences (No.\ XDB0510201).

H.\ Huang was partially supported by the NSFC grant (No.\ 12101105) and the Natural Science Foundation of Fujian 
Province of China (No.\ 2024J01271).
	

\newcommand{\Gathen}{\relax}\newcommand{\Hoeij}{\relax}\newcommand{\Hoeven}{\relax}\def\cprime{$'$}
\def\cprime{$'$} \def\cprime{$'$} \def\cprime{$'$} \def\cprime{$'$}
\def\cprime{$'$} \def\cprime{$'$} \def\cprime{$'$} \def\cprime{$'$}
\def\polhk#1{\setbox0=\hbox{#1}{\ooalign{\hidewidth
			\lower1.5ex\hbox{`}\hidewidth\crcr\unhbox0}}} \def\cprime{$'$}

\end{document}